\documentclass[sloppy,oribibl,envcountsame,envcountsect,11pt]{article}

\usepackage[utf8]{inputenc} 
\usepackage{amsmath,amssymb}

\usepackage{a4}

\usepackage{amsfonts,amssymb,amsthm,amsmath,comment}
\usepackage{hyperref}

\usepackage{newpxtext}

\usepackage[varg,bigdelims]{newpxmath}
\usepackage[margin=1.5in]{geometry}

\usepackage{float}
\usepackage{comment}

\usepackage{todonotes}
\newcommand{\mnote}[3]{\todo[color=#3!40,size=\footnotesize]{\textbf{#2:} #1}}

\newcommand{\pushkar}[1]{\mnote{#1}{P}{blue}}

\title{A Multivariate to Bivariate Reduction for Noncommutative
  Rank and Related Results}

\author{V. Arvind\thanks{The Institute of Mathematical Sciences (HBNI),
    Chennai, India,\texttt{  email: arvind@imsc.res.in} and Chennai Mathematical Institute, Siruseri, Kelambakkam, India} \and Pushkar S
  Joglekar\thanks{Vishwakarma Institute of Technology, Pune, India,
    \texttt{email: joglekar.pushkar@gmail.com}}}

\date{}

\newtheorem{theorem}{Theorem}[section]

\newtheorem{definition}[theorem]{Definition}
\newtheorem{lemma}[theorem]{Lemma}
\newtheorem{fact}[theorem]{Fact}
\newtheorem{claim}[theorem]{Claim}
\newtheorem{remark}[theorem]{Remark}

\newtheorem{proposition}[theorem]{Proposition}

\def\qed{\hspace*{\fill} $\Box$\par\medskip}

\DeclareMathOperator{\ncrk}{\rm{ncrk}}

\usepackage{graphicx}

\usepackage{amsmath,amssymb}

\newcommand{\op}[1]{\ensuremath{\operatorname{#1}}}

\newcommand{\fR}{\FX}

\newcommand{\F}{\mathbb{F}}
\newcommand{\N}{\mathbb{N}}
\renewcommand{\angle}[1]{\langle #1 \rangle}
\newcommand{\FX}{\F\angle{X}}
\newcommand{\FZ}{\F\angle{Z}}

\newcommand{\M}{\mathbb{M}}

\newcommand{\skewf}{\mathbb{F}\newbrak{X}}
\newcommand{\fF}{\skewf}
\newcommand{\fFxy}{\F\newbrak{x,y}}

\DeclareFontEncoding{LS1}{}{}
\DeclareFontSubstitution{LS1}{stix}{m}{n}
\DeclareSymbolFont{symbols2stix}{LS1}{stixfrak} {m} {n}
\DeclareMathSymbol{\lparenless}{\mathopen} {symbols2stix}{"32}
\DeclareMathSymbol{\rparengtr}{\mathclose}{symbols2stix}{"33}

\newcommand{\newbrak}[1]{{\lparenless} #1 {\rparengtr}}
\newcommand{\ratX}{\mathbb{F}\newbrak{X}}

\newcommand{\NC}{\mathsf{NC}}

\newcommand{\poly}{\op{poly}}
\newcommand{\polylog}{\op{polylog}}

\newcommand{\RIT}{\textrm{RIT}}
\newcommand{\NCRANK}{\textrm{ncRANK}}
\newcommand{\pNCRANK}{\NCRANK_{poly}}
\newcommand{\SDIT}{\textrm{SDIT}}
\newcommand{\SING}{\textrm{SINGULAR}}

\usepackage{framed}

\begin{document}

\maketitle

\begin{abstract}
  We study the \emph{noncommutative rank} problem, $\NCRANK$, of
  computing the rank of matrices with linear entries in $n$
  noncommuting variables and the problem of \emph{noncommutative
    Rational Identity Testing}, $\RIT$, which is to decide if a given
  rational formula in $n$ noncommuting variables is zero on its domain
  of definition.

  Motivated by the question whether these problems have
  \emph{deterministic} $\NC$ algorithms, we revisit their
  interrelationship from a parallel complexity point of view.
  We show the following results:
  \begin{enumerate}
  \item Based on Cohn's embedding theorem \cite{Co90,Cohnfir}
    we show deterministic $\NC$ reductions from multivariate
    $\NCRANK$ to bivariate $\NCRANK$ and from multivariate $\RIT$
    to bivariate $\RIT$.
  \item We obtain a deterministic $\NC$-Turing reduction from
    bivariate $\RIT$ to bivariate $\NCRANK$, thereby proving that
    a deterministic $\NC$ algorithm for bivariate $\NCRANK$
    would imply that both multivariate $\RIT$ and multivariate
    $\NCRANK$ are in deterministic $\NC$.
 \end{enumerate}

\end{abstract}

\section{Introduction}

There are two main algorithmic problems of interest in this
paper. These are the \emph{noncommutative Rational Identity Testing
  problem} ($\RIT$) and the \emph{noncommutative rank} ($\NCRANK$)
problem for matrices with linear entries.

The $\RIT$ problem is a generalization of multivariate polynomial
identity testing to identity testing of multivariate rational
expressions. When the variables are commuting, rational identity
testing and polynomial identity testing are equivalent problems. On
the other hand, if the variables are all noncommuting, the $\RIT$
problem needs different algorithmic techniques as rational expressions
in noncommuting variables are more complicated. Mathematically,
rational expressions over noncommuting variables are quite well
studied.  They arise in the construction of the so-called free skew
fields \cite{Cohnfir}. Hrubes and Wigderson \cite{HW15} initiated the
algorithmic study of $\RIT$ for \emph{rational
  formulas} and gave a deterministic polynomial time reduction from $\RIT$ to $\NCRANK$. Subsequently, deterministic polynomial-time algorithms
were obtained independently by Ivanyos et al \cite{IQS18,IQS17} and by
Garg et al \cite{GGOW15, GGOW20} for the $\RIT$ problem, in fact they obtain deterministic polynomial time algorithms for $\NCRANK$, and using Hrubes-Wigderson reduction from $\RIT$ to $\NCRANK$ get a polynomial time algorithm for $\RIT$. The Ivanyos et al algorithm is algebraic and works for
fields of all characteristics. The Garg et al algorithm has an
analytic flavor and is for the characteristic zero case.

\paragraph{The Edmonds' Problem and $\NCRANK$} 
The $\NCRANK$ problem is essentially the noncommutative version of the
well-known Edmonds' problem: determine the rank of a matrix $M$ whose
entries are linear forms in commuting variables (see
\cite{IQS18,GGOW20,BJP,CM23} for more details). A special case of it
is to determine if a square matrix $M$ with linear entries in
commuting variables is \emph{singular}. This is also known as the
symbolic determinant identity testing problem, $\SDIT$. There is an
easy randomized $\NC$ algorithm for it, based on the Polynomial
Identity Lemma \cite{AJMR19,Sch80,Zippel79,DL78}, by randomly
substituting scalar values for the variables from the field (or a
suitable extension of it) and evaluating the determinant using a
standard $\NC$ algorithm.  However, a deterministic polynomial-time
algorithm for $\SDIT$ is an outstanding open problem \cite{BJP}.

Recently, for the $\RIT$ problem the first deterministic quasi-$\NC$
algorithm has been obtained \cite{ACM24a}. Another recent development
--- building on the connection between the noncommutative Edmonds'
problem and identity testing for noncommutative algebraic branching
programs \cite{CM23} --- is a generalization of the Edmond's problem
to a partially commutative setting with application to the weighted
$k$-tape automata equivalence problem \cite{ACM24b}.

\subsection{This paper: overview of results and proofs}

With this background, the natural algorithmic questions are whether
$\RIT$ for noncommutative rational formulas and $\NCRANK$ have
deterministic $\NC$ algorithms. We revisit the problems from this
perspective and obtain the following new results.



\begin{itemize}
\item[1.]~ We show that multivariate $\RIT$ for formulas is
  deterministic $\NC$ reducible to bivariate $\RIT$ for formulas. More
  precisely, given a rational formula $\Phi(x_1,x_2,\ldots,x_n)$,
  computing an element of the skew field $\ratX$, where
  $X=\{x_1,x_2,\dots,x_n\}$, the deterministic NC reduction replaces
  each $x_i$ by a formula $\Phi_i(x,y)$ computing a polynomial in
  $\F\angle{x,y}$. Then the resulting rational formula
  \[
  \Psi(x,y) = \Phi(\Phi_1(x,y),\Phi_2(x,y),\dots,\Phi_n(x,y))
  \]
  has the property that
  \[
    \Phi(x_1,x_2,\dots,x_n)\ne 0 \textrm{ iff } \Psi(x,y)\ne 0.
  \]

\item[2.]~ We next show that multivariate $\NCRANK$ is deterministic
  $\NC$ reducible to bivariate $\NCRANK$. More precisely, given a
  $d\times d$ linear matrix $A=A_0+\sum_{i=1}^n A_i x_i$ in
  noncommuting variables $X=\{x_1,x_2,\ldots,x_n\}$, where the $A_i$
  are matrices over the scalar field $\F$, we first give a
  deterministic NC reduction that transforms $A$ to a $d\times d$ matrix $B$
  whose entries are bivariate polynomials in $\F\angle{x,y}$, where
  $x$ and $y$ are two noncommuting variables, where its entries
  $B[i,j]$ are given by \emph{polynomial size noncommutative
    formulas}, with the property that $\ncrk(A)=\ncrk(B)$. Then we
  examine the Higman linearization process \cite{HW15} that transforms
  $B$ into a matrix $B'$ with linear entries in $x$ and $y$ such that
  the noncommutative rank $\ncrk(B)$ of $B$ can be easily recovered
  from $\ncrk(B')$. We show that this process can be implemented in
  deterministic $\NC$ (the earlier works
  \cite{HW15,IQS17,IQS18,GGOW20} only consider its polynomial-time
  computability).

  
  Additionally, we consider the more general problem $\pNCRANK$ of
  computing the noncommutative rank of a matrix whose entries are
  noncommutative formulas computing polynomials. We show using our
  parallel Higman linearization algorithm that $\pNCRANK$ is also
  deterministic $\NC$ reducible to bivariate $\NCRANK$.

  Both the multivariate to bivariate reductions, stated above, are
  crucially based on a theorem of Cohn \cite{Co90} (also see 
  \cite[Theorem 4.7.3]{Cohnfir}) which we will refer to as Cohn's 
  embedding theorem and describe it later in the introduction.

\item[3.]~ Finally, obtaining a deterministic $\NC$ reduction from
  $\RIT$ to $\NCRANK$ turns out to be quite subtle. From the work of
  Hrubes and Wigderson \cite{HW15}, who initiated this line of
  research on $\RIT$, we can only obtain a sequential deterministic
  polynomial-time reduction from $\RIT$ to $\NCRANK$. However, for
  our result we require an $\NC$ reduction. If the given rational
  formula has logarithmic depth, then their result already implies an
  $\NC$ reduction.

  Now, in the same paper \cite{HW15}, Hrubes and Wigderson have also
  shown a \emph{depth reduction} result for multivariate
  noncommutative rational formulas: every rational formulas of size
  $s$ is equivalent to a logarithmic depth rational formula of size
  $\poly(s)$. Their construction is based on Brent's depth reduction
  result for commutative arithmetic formulas. 
  However, due to noncommutativity and the presence of inversion
  gates, the formula constructed in their proof needs to be different
  based on whether certain rational subformulas, arising in the
  construction procedure, are identically zero or not. To algorithmize
  such steps in the construction we need to use $\RIT$ as a
  subroutine. As $\RIT$ has a
  polynomial-time algorithm \cite{IQS18,GGOW20}, the depth-reduction
  in \cite{HW15} also has a polynomial time algorithm.\footnote{In the
    commutative case, Brent's result is parallelizable to yield an
    $\NC$ algorithm. For noncommutative formulas without inversion
    gates we can obtain the depth-reduced formula in $\NC$, as we will
    observe later in the paper.}

  As the third result of this paper, building on the Hrubes-Wigderson
  depth-reduction construction, we are able to show that, with oracle
  access to $\RIT$, rational formula depth reduction can be done in
  deterministic $\NC$. Using this we are able to obtain a
  deterministic $\NC$-Turing reduction from $\RIT$ to $\NCRANK$.
  Hence, if bivariate $\NCRANK$ is in deterministic $\NC$ we will
  obtain a deterministic $\NC$ algorithm also for $\RIT$. We leave
  open the question whether depth reduction of noncommutative rational
  formulas is unconditionally in $\NC$.


\end{itemize}

We now outline the proof ideas in more detail with some intuitive
explanations.

\section{Preliminaries}\label{prel-sec}

In this section we recall the essential basic definitions and fix the
notation.

Let $\F$ be a (commutative) field\footnote{In this paper, $\F$ will
either be the field of rationals or a finite field.} and
$X=\{x_1,x_2,\ldots,x_n\}$ be $n$ free noncommuting variables. The
free monoid $X^*$ is the set of all monomials in the variables $X$. A
\emph{noncommutative polynomial} $f(X)$ is a finite $\F$-linear
combination of monomials in $X^*$, and the \emph{free noncommutative
ring} $\FX$ consists of all noncommutative polynomials.

\paragraph{Noncommutative Rational Formulas} An \emph{arithmetic
  circuit} computing an element of $\FX$ is a directed acyclic graph
with each indegree $0$ node labeled by either an input variable
$x_i\in X$ or some scalar $c\in\F$. Each internal node $g$ has
indegree $2$ and is either a $+$ gate or a $\times$ gate: it computes
the sum (resp.\ left to right product) of its inputs. Thus, each gate
of the circuit computes a polynomial in $\FX$ and the polynomial
computed by the circuit is the polynomial computed at the \emph{output
  gate}. A \emph{formula} is restricted to have fanout $1$ or $0$.

When we allow the formulas/circuits to have \emph{inversion gates} we
get \emph{rational formulas and rational circuits}.


\paragraph{The Free Skew Field}

We now briefly explain the free skew field construction. 
The elements of the free skew field are noncommutative rational
functions which are more complicated than their commutative
counterparts. Rational formulas in the commutative setting can be
canonically expressed as ratios of two polynomials.  There is no such
canonical representation for noncommutative rational
formulas.

Following
Hrubes-Wigderson \cite{HW15}, we use Amitsur's approach \cite{Am66}
for formally defining skew fields.\footnote{There are other ways to defining free
skew fields \cite{Am66,Be70,Ro80,Co71,Co72,Cohnrelations,Cohnfir}.}


It involves defining appropriate notion of equivalence of formulas (intuitively, two formulas are equivalent if they agree on their \emph{domain of definition}). The equivalence classes under this equivalence relation are the elements of the free skew field. We give the formal definitions below.


Let $\M_k(\F)$ denote the ring of $k\times k$ matrices
with entries from field $\F$. Note that a rational formula $\Phi$
defines a partial function
\[
\hat{\Phi}:\M_k(\F)^n \mapsto \M_k(\F)
\]
that on input $(a_1, a_2, \ldots, a_n)\in \M_k(\F)^n$ evaluates $\Phi$
by substituting $x_i \leftarrow a_i$ for $i\in [n]$.
$\hat{\Phi}(a_1, \ldots, a_n)$ is undefined if the input to some
inversion gate in $\Phi$ is not invertible in $\M_k(\F)$.

\begin{definition}
  Let $\Phi$ be a rational formula in variables $X$. For each
  $k\in\N$, let $\mathcal{D}_{k,\Phi}$ be the set of all matrix tuples
  $(a_1, a_2, \ldots, a_n)\in\M_k(\F)^n$ such that $\hat{\Phi}(a_1,
  a_2, \ldots, a_n)$ is defined. The \emph{domain of definition} of
  $\Phi$ is the union $\mathcal{D}_\Phi=\bigcup_k \mathcal{D}_{k,\Phi}$.
\end{definition}

\begin{definition}{\rm\cite{HW15}}\hfill{~}\label{correct-def}
\begin{itemize}
\item A rational formula $\Phi$ is called \emph{correct} if for every
  gate $u$ of $\Phi$ the subformula $\Phi_u$ has a nonempty domain of
  definition.
\item Correct rational formulas $\Phi_1,\Phi_2$ are said to be
  \emph{equivalent} (denoted $\Phi_1 \equiv \Phi_2$) if the
  intersection $\mathcal{D}_{\Phi_1}\cap\mathcal{D}_{\Phi_2}$ of their
  domains of definitions is nonempty and they agree on all the points
  in the intersection.

\end{itemize}  
\end{definition}


We note that equivalent formulas need not have the same domain of
definition. For example, $\Phi_1 = z_1z_2z_3$ and $\Phi_2=(z_1z_2z_3
\cdot (z_2z_3 - z_3z_2)^{-1})\cdot (z_2z_3 - z_3z_2)$ are
equivalent. However, the domain of definition of $\Phi_1$ includes all
matrix tuples, whereas the domain of definition of $\Phi_2$ contains
only matrix tuples $(Z_1,Z_2, Z_3)$ such that $Det(Z_2 Z_3 - Z_3 Z_2)
\neq 0$.


The relation $\equiv$ as defined above is an equivalence relation on
rational formulas and the equivalence classes, called \emph{rational
functions}, are the elements of the skew field $\fF$ \cite{HW15,Am66}.

\paragraph{Noncommutative Rank}

We now recall the notion of rank for matrices over the noncommutative
ring $\FX$. 

\begin{definition}[inner rank]\label{inner}
  Let $M$ be a matrix over $\FX$. Its \emph{inner rank} is the least
  $r$ such that $M$ can be written as a matrix product $M=PQ$ where
  $Q$ has $r$ rows (and $P$ has $r$ columns).
\end{definition}

\begin{definition}[full matrices]
  An $n\times n$ square matrix $M$ over $\FX$ is \emph{full} if it
  cannot be decomposed as a matrix product $M=PQ$ where $P$ is
  $n\times r$ and $Q$ is $r\times n$ for $r<n$. In other words, an
  $n\times n$ matrix is called full precisely when its inner rank is
  $n$.
\end{definition}

We can also define the rank of a matrix $M$ to be the maximum $r$ such
that $M$ contains an $r\times r$ full submatrix. For matrices over
$\FX$ these notions of \emph{noncommutative rank} coincide as
summarized below.\footnote{For a ring $R$ in general, a full matrix
  $R$ need not be invertible (see \cite{HW15} for an example).}

\begin{proposition}{\rm\cite{Cohnfir}}\hfill{~}
  
  Let $M$ be an $n\times n$ matrix over the ring $\FX$. Then
  \begin{itemize}
  \item $M$ is a full matrix (that is, $M$ has inner rank $n$) iff it is invertible over the skew field
    $\ratX$.
  \item More generally, $M$ has inner rank $r$ iff the largest full
    submatrix of $M$ is $r\times r$.
  \end{itemize}
\end{proposition}

\paragraph{The Algorithmic Problems of Interest}

\vspace{2mm}

At this point we formally define the problems of interest in
this paper.
\begin{enumerate}
\item The multivariate $\RIT$ problem takes as input a rational
  formula $\Phi$, computing a rational function $\hat{\Phi}$ in $\fF$,
  and the problem is to check if $\Phi$ is equivalent to $0$?
    In the bivariate $\RIT$ problem $\Phi$ computes 
  a rational function in $\fFxy$.


\item The multivariate $\NCRANK$ problem takes as input a matrix
  $M$ with affine linear form entries over $X$ and the problem is
  to determine its noncommutative rank $\ncrk(M)$. Bivariate $\NCRANK$
  is similarly defined.
\item A more general version of $\NCRANK$ is $\pNCRANK$ in which the
  matrix entries are allowed to be polynomials in $\FX$ computed by
  noncommutative formulas. A closely related problem is $\SING$ where
  the problem is to test if a square matrix $M$ over $\FX$ with
  entries computed by formulas is singular or not.
\end{enumerate}

\paragraph{The complexity class $\NC$ and $\NC$ reductions}
The class $\NC$ consists of decision problems that can be
solved in $\polylog(n)$ time with $\poly(n)$ many processors.\footnote{
This model is widely accepted as the right theoretical notion for efficient
parallel algorithms.} For two decision problems $A$ and $B$
we say $A$ is many-one $\NC$ reducible to $B$ if there is
a reduction from $A$ to $B$ that is $\NC$ computable. Similarly,
$A$ is $\NC$-Turing reducible to $B$ if there is an oracle
$\NC$ algorithm for $A$ that has oracle access to $B$.

  It turns out that $\SING$ and $\NCRANK$ are equivalent even under
  deterministic $\NC$ reductions. \footnote{As for $M \in \FX^{m\times n}$, $\ncrk(M)=r$ iff $r$ is a size of a largest sized invertible minor of $M$, so to compute $\ncrk(M)$, it suffices to test singularity of matrix $UMV$, where $U,V$ are generic $r\times m, n \times r$ matrices respectively with entries as fresh noncommuting variables for $r\leq \min(m,n)$. See e.g. \cite[Lemma A.3]{GGOW20} for details.}

\paragraph{Cohn's Embedding Theorem}

\vspace{2mm}

We now give an outline of Cohn's embedding theorem and how
it gives us the desired reduction from multivariate to
bivariate $\RIT$ and also from multivariate to bivariate
$\NCRANK$. However, for multivariate to bivariate reduction for $\NCRANK$ we 
will require additional $\NC$ algorithms for formula depth reduction and Higman linearization. 

Let $X=\{x_1, x_2,\ldots,x_n\}$ be a set of $n$ noncommuting variables, and let
$x, y$ be a pair of noncommuting variables. We first recall the
following well-known fact, observed in the early papers on
noncommutative polynomial identity testing \cite{BW05,RS05}: for
noncommutative polynomials in $\FX$, the problem of  polynomial identity
testing (PIT) is easily reducible to PIT for bivariate noncommutative
polynomials in $\F\angle{x,y}$. Indeed, more formally, we have the
following easy to check fact.

\begin{fact}
The map
\[
x_i\mapsto x^{i-1}y, 1\le i\le n
\]
extends to an injective homomorphism (i.e.\ a \emph{homomorphic
  embedding}) from the ring $\FX$ to the ring $\F\angle{x,y}$.
\end{fact}

However, in order to obtain our multivariate to bivariate reductions,
we need a mapping $\beta:X\to \F\angle{x,y}$ which has the following
properties:
  \begin{itemize}
  \item For each $i$, there is a small noncommutative
    arithmetic formula that computes $\beta(x_i)$.
  \item $\beta$ extends to an injective homomorphism\footnote{That is,
      a homomorphic embedding.}, not just from the ring $\FX$ to
    $\F\angle{x,y}$, but also to an injective homomorphism from the
    skew field $\fF$ to the skew field $\fFxy$. This will guarantee
    that for two rational formulas $\Phi_1,\Phi_2$ computing
    inequivalent rational functions in $\fF$ their images
    $\beta(\Phi_1)$ and $\beta(\Phi_2)$ also compute inequivalent
    rational functions in $\F\angle{x,y}$. 
  \item Furthermore, in order to get the multivariate to bivariate
    reduction for $\NCRANK$, we will additionally require of the map
    $\beta$ that for any matrix $M$ over $\FX$ its image $\beta(M)$,
    which is a matrix over $\F\angle{x,y}$ obtained by applying
    $\beta$ to each entry of $M$, has the same rank as $M$. Such a
    homomorphic embedding is called an \emph{honest embedding}
    \cite{Co90}. Here we note that, full matrices over $\FX$ are invertible over $\fF$ \cite{Cohnfir}. Consequently if one can lift embedding $\beta$ to one between the corresponding free skew fields, it enforces $\beta$ to be an honest embedding.  
  \end{itemize}

  The mapping $x_i\mapsto x^{i-1}y$ actually \emph{does not} extend to an
  honest embedding as observed in \cite{Co90}. Indeed, the rank $2$
  matrix $\left(
\begin{array}{cc}
x_1 & x_2  \\
x_3 & x_4
\end{array}
\right)$
has image
\[
\left(
\begin{array}{cc}
y & xy \\
x^2y & x^3y
\end{array}
\right)=\left(\begin{array}{c}1\\x^2\end{array}\right)
\left(\begin{array}{cc}y & xy\end{array}\right)
\]
which is rank $1$. In general, a homomorphic embedding from a ring $R$
to a ring $R'$ is an \emph{honest embedding} if it maps full matrices
over $R$ to full matrices over $R'$. We now state Cohn's embedding
theorem.

For polynomials $f,g\in \F\angle{x,y}$ let $[f,g]$ denotes the
commutator polynomial $fg-gf$. Cohn's embedding map
$\beta:\FX\to \F\angle{x,y}$ is defined as follows.
\begin{itemize}
\item Let $\beta(x_1)=y$. For $i\geq 2$, define
  $\beta(x_i)= [\beta(x_{i-1}),x]$.
\item We can then naturally extend $\beta$ to a homomorphism from
  $\FX$ to $\F\angle{y,x}$, and it is easy to check that it is
  injective. In fact, we can even assume $|X|$ to be countably
  infinite.
\end{itemize}  

\begin{theorem}[Cohn's embedding theorem]{\rm\cite[Theorem 7.5.19]{Cohnfir}}\label{cohn-thm}
  The embedding map $\beta:\FX\to \F\angle{x,y}$ defined above extends to an embedding between the corresponding skew fields
  $\beta : \F\newbrak{Z}\to \F\newbrak{x,y}$ and hence is an honest
  embedding.
\end{theorem}  
Cohn's construction is based on skew polynomial rings, which explains
the appearance of the iterated commutators
$\beta(x_i)=[\beta(x_{i-1}),x]$. We briefly sketch the underlying ideas
in the appendix (see Section~\ref{cohn-sec}). For more details see
\cite{Co90,Cohnfir}.


\section{The Reduction from multivariate $\RIT$ to bivariate $\RIT$}

The reduction follows quite easily from Theorem~\ref{cohn-thm}. However,
we present some complexity details in this section showing that it is
actually a deterministic $\NC$ reduction. The following lemma is
useful to describe the reduction.

\begin{lemma}
  Recall the embedding map $\beta$ defined above. $\beta(z_0)=y$ and
  $\beta(z_{i+1})=[\beta(z_i),x]$ are polynomials in $\F\angle{x,y}$
  for each $i\ge 0$. Then, for $n\geq 1$ we have
  \[\beta(z_n) = \sum_{i=0}^{n} (-1)^i {n \choose i} x^iyx^{n-i}.
  \]
  As a consequence, there is a deterministic $\NC$ algorithm that
  constructs a $\poly(n)$-sized formula for $\beta(z_n)$. 

\end{lemma}

\begin{proof}
We will use induction on $n$. The base case follows from the fact that
$\beta(z_1)= yx - xy$. Inductively assume the claim is true for all
$n$. Now, $\beta(z_{n+1})= [\beta(z_{n}), x]$
\begin{eqnarray*}
&=& \sum_{i=0}^n (-1)^i {n \choose i}x^i y x^{n-i+1}- \sum_{j=0}^n (-1)^j {n \choose j}x^{j+1} y x^{n-j} ~ \text{by induction hypothesis}\\
&=& yx^{n+1}+ \sum_{i=1}^n (-1)^i {n \choose i}x^i y x^{n-i+1} + \sum_{i=1}^{n+1} (-1)^i {n+1-1 \choose i-1}x^{i} y x^{n+1-i}\\
&=& yx^{n+1}+ \sum_{i=1}^n (-1)^i ~~\left[{n+1-1 \choose i} + {n+1-1 \choose i-1} \right]~~ x^i y x^{n-i+1} + (-1)^{n+1} x^{n+1}y \\
&=& yx^{n+1} + (-1)^{n+1}x^{n+1} y + \sum_{i=1}^{n} (-1)^i {n+1 \choose i} x^i yx^{n+1-i}~~\text{from Pascal's identity}\\
&=& \sum_{i=0}^{n+1} (-1)^i {n+1 \choose i} x^i yx^{n+1-i}
\end{eqnarray*}
This completes the inductive proof.

As the binomial coefficients can be computed in $\NC$ using Pascal's identity, the expression for $\beta(z_n)$ obtained above immediately implies an $\NC$ algorithm for construction of a $\poly(n)$ sized formula for $\beta(z_n)$. 
\end{proof}


\begin{theorem}\label{thm-n-rit-2-rit}
  The multivariate $\RIT$ problem is deterministic $\NC$ (in fact,
  logspace) reducible to bivariate $\RIT$. More precisely, given as
  input a rational formula $\Phi$ computing an element of $\ratX$,
  $X=\{x_1,x_2,\ldots,x_n\}$ there is a deterministic $\NC$
  algorithm that computes a rational formula $\Psi$ computing an
  element of $\F\newbrak{x,y}$ such that $\Phi$ is nonzero in its
  domain of definition iff $\Psi$ is nonzero in its domain of
  definition.
\end{theorem}

\begin{proof}
  We can identify $\ratX$ with $\F\newbrak{z_0,z_1,\ldots,z_{n-1}}$.
  Let $\Phi_i(x,y)$ be the $\poly(i)$ size noncommutative formula
  computing the nested commutator $\beta(z_i)$ for each $i$. In the
  rational formula $\Phi$, for each $i$ we replace the input $z_i$ to
  $\Phi$ by $\Phi_i(x,y)$. The new formula we obtain is
  \[
  \Psi(x,y) = \Phi(\Phi_1(x,y),\Phi_2(x,y),\ldots,\Phi_{n-1}(x,y)).
  \]
  By Theorem~\ref{cohn-thm}, $\Psi(x,y)\ne 0$ on its domain of
  definition iff $\Phi(z_0,z_1,\ldots,z_{n-1})$ is nonzero on its
  domain of definition. Furthermore, because $\beta$ is an
  embedding, it is guaranteed that if $\Phi$ has a nontrivial
  domain of definition then $\Psi$ also has a nontrivial
  domain of definition.

  As computation of $\Psi$ from $\Phi$ involves only replacing the
  $z_i$ by $\Phi_i$, the reduction is clearly logspace computable.
\end{proof}

\section{Reduction from $n$-variate $\pNCRANK$ to $2$-variate
  $\NCRANK$}

In this section we give a deterministic $\NC$ reduction from
$n$-variate $\pNCRANK$ to bivariate $\NCRANK$. The basic idea of the
reduction is as follows. Given a polynomial matrix\footnote{We can
  assume it is a square matrix without loss of generality.}
$M\in \FX^{m\times m}$ such that each entry of $M$ is computed by
formula of size at most $s$. We will use Cohn's embedding theorem
\ref{cohn-thm} to get a matrix $M_1$ with \emph{bivariate} polynomial
entries such that each entry of $M_1$ is computed by a $\poly(n,s)$
size noncommutative formula and $\ncrk(M)=\ncrk(M_1)$. Notice that
$M_1$ is an instance of bivariate $\pNCRANK$. Next, we need to give an
$\NC$ reduction transforming $M_1$ to an instance of bivariate
$\NCRANK$ (which will be a matrix with linear entries in $x$ and $y$).

In order to do this transformation in $\NC$, we will first apply the
depth-reduction algorithm of Lemma
\ref{lem-simple-formula-depth-reduction} to get matrix $M_2$ whose
entries are $\poly(n,s)$ size log-depth formulas that compute the same
polynomials as the corresponding entries of $M_1$. Then we apply
Higman Linearization to $M_2$ to obtain a bivariate linear matrix
$M_3$. For this we will use our parallel algorithm for Higman
Linearization described in Theorem
\ref{thm-parallel-effective-Higman}. From the properties of Higman
linearization we can easily recover $\ncrk(M_2)$ from
$\ncrk(M_3)$.  In what follows, first we give a deterministic $\NC$
algorithm for the depth reduction of noncommutative formulas and
Higman linearization process. We conclude the section by giving an
$\NC$ reduction from multivariate to bivariate $\NCRANK$ using above
$NC$ algorithms combined with Cohn's embedding theorem \ref{cohn-thm}.

\subsection{Depth reduction for noncommutative formulas without divisions}

In the commutative setting Brent \cite{Br74} obtained a deterministic
$\NC$ algorithm to transform a given \emph{rational} formula (which
may have division gates) to a log-depth rational formula. In the
noncommutative setting, Hrubes and Wigderson \cite{HW15} proved the
\emph{existence} of log-depth \emph{rational} formula equivalent to
any given rational formula.  Their proof is based on \cite{Br74}.
However, it is not directly algorithmic as explained in the
introduction. We will discuss it in more detail in Section
\ref{sec-rit-ncrank-reduction}. However, it turns out that, if the
noncommutative formula doesn't have division gates then the depth
reduction is quite easy and we obtain a simple deterministic $\NC$
algorithm for it that computes a log-depth noncommutative formula
equivalent to the given noncommutative formula. The proof is based
Brent's commutative version. We just highlight the distinctive points
arising in the noncommutative version in the proof presented in the
appendix.
    
We introducing some notation. Let $\Phi$ be a noncommutative
arithmetic formula computing a polynomial in
$\F\langle x_1, x_2, \ldots, x_n \rangle$. Let $\hat{\Phi}$ denote the
polynomial computed by $\Phi$. For a node $v \in \Phi$, let $\Phi_v$
denote the subformula of $\Phi$ rooted at node $v$, so $\hat{\Phi_v}$
is the polynomial computed by the subformula rooted at $v$.
For a node $v\in \Phi$, let $\Phi_{v\leftarrow z}$ be a formula obtained from $\Phi$ by replacing the sub-formula $\Phi_v$ by single variable $z$. For a node $v\in \Phi$, let $wt(v)= |\Phi_v|$ denote the number of nodes in the subformula rooted at $v$. By size of formula $\Phi$ we refer to number of gates in $\Phi$.  

\begin{lemma}\label{lem-simple-formula-depth-reduction}
  Given a formula $\Phi$ of size $s$ computing a noncommutative
  polynomial $f \in \F\langle x_1, x_2, \ldots, x_n\rangle$ there is
  an $\NC$ algorithm to obtain an equivalent formula $\Phi'$ for $f$
  with depth $O(\log s)$.
\end{lemma}

The proof of the lemma is in the appendix.


\begin{remark}
  In Lemma~\label{lem-simple-formula-depth-reduction}, we avoid using
  the depth-reduction approach for noncommutative rational formulas in
  \cite{HW15}. This is because it can introduce inversion gates even
  if the original formula has no inversion gates. Instead, we directly
  adapt ideas from Brent's construction for commutative formulas to
  obtain the NC algorithm. We note here that Nisan's seminal work
  \cite{Ni91} also briefly mentions noncommutative formula depth
  reduction (but not its parallel complexity or even in an algorithmic
  context).
\end{remark} 


Higman linearization which is sometimes called Higman's trick was
first used by Higman in \cite{Hig}. Cohn extensively used Higman
Linearization in his factorization theory of free ideal rings. Given a
matrix with noncommutative polynomials as its entries, Higman
linearization process transforms it into a matrix with linear
entries. This transformation process has several nice properties such
as: it preserves fullness of the matrix (that is the input polynomial
matrix is full rank iff final linear matrix is full rank), it
preserves irreducibility of the matrix, etc.

We first describe a single step of the linearization process applied
to a single noncommutative polynomial, which easily generalizes to
matrices with polynomial entries. Given an $m \times m$ matrix $M$
over $\fR$ such that $M[m,m]=f+g\times h$, apply the following:
\begin{itemize}
\item Expand $M$ to an  $(m+1)\times (m+1)$ matrix by adding a new last
  row and last column with diagonal entry $1$ and remaining new entries
  zero:
\[
\left[
\begin{array}{c|c}
M & 0 \\
\hline
0 & 1
\end{array}
\right]. 
\]
\item Then the bottom right $2\times 2$ submatrix is transformed as
  follows by elementary row and column operations

\[
\left(
\begin{array}{cc}
f+gh & 0 \\
0 & 1
\end{array}
\right)\rightarrow
\left(
\begin{array}{cc}
f+gh & g \\
0 & 1
\end{array}
\right)\rightarrow
\left(
\begin{array}{cc}
f & g \\
-h & 1
\end{array}
\right)
\]

\end{itemize}

Given a polynomial $f\in\FX$ by repeated application of the above step
we will finally obtain a \emph{linear matrix} $L=A_0+\sum_{i=1}^n A_i
x_i$, where each $A_i, 0\le i\le n$ is an $\ell\times \ell$ over $\F$,
for some $\ell$. The following theorem summarizes its properties.

\begin{theorem}[Higman Linearization]{\rm\cite{Cohnfir}}\label{higthm}
  Given a polynomial $f\in\fR$, there are matrices $P, Q\in\fR^{\ell \times \ell}$
  and a linear matrix $L\in\fR^{\ell\times \ell}$ such that
\begin{equation}\label{higeq}
P \left(
\begin{array}{c|c}
f & 0 \\
\hline
0 & I_{\ell-1}
\end{array}
\right) ~Q=~L
\end{equation}
with $P$ upper triangular, $Q$ lower triangular, and the diagonal
entries of both $P$ and $Q$ are all $1$'s. (Hence, $P$ and $Q$ are both
invertible over $\fF$, moreover entries of $P^{-1}$ and $Q^{-1}$ are in $\FX$).
\end{theorem}
Instead of a single $f$, we can apply Higman linearization to a matrix
of polynomials $M\in\fR^{m\times m}$ to obtain a linear matrix $L$
such that $P(M \oplus I_k)Q = L$ for invertible matrices $P,Q$.  Garg et
al.~\cite{GGOW20} gave polynomial time algorithm to carry out Higman
linearization for polynomial matrix whose entries are given by
noncommutative formulas. We will give an $\NC$ algorithm to implement
this transformation.

\begin{theorem}\label{thm-parallel-effective-Higman}
Let $A \in \FX^{n \times n}$ be a polynomial matrix such that each entry of $A$ is given by a noncommutative formula of size at most $s$. Then there is a deterministic $\NC$ algorithm (with parallel time complexity $\poly( \log s, \log n)$) to compute invertible upper and lower triangular matrices $P, Q \in \FX^{\ell \times \ell}$ with all diagonal entries $1$ and a linear matrix $L \in \FX^{\ell \times \ell}$ such that $ P ~\left(
\begin{array}{c|c}
A & 0 \\
\hline
0 & I_{k}
\end{array}
\right)~Q = L$, where $\ell = n+k$ and $k$ is $O(n^2\cdot s)$. 
All the entries of $P,Q$ are computable by algebraic branching programs of size $\poly(n,s)$. 
Moreover $ncrk(A)+k = ncrk(L)$, hence the rank of $A$ is easily computable from the rank of $L$.  
\end{theorem}

The proof of the above theorem is in the appendix.

Using Theorem \ref{cohn-thm}, Lemma
\ref{lem-simple-formula-depth-reduction} and Theorem
\ref{thm-parallel-effective-Higman} we get a deterministic $\NC$ reduction
from multivariate $\NCRANK$ to bivariate $\NCRANK$.

\begin{theorem}
There is a deterministic $\NC$ reduction from the multivariate $\NCRANK$ problem to the bivariate $\NCRANK$ problem.
\end{theorem} 


\section{$\NC$ Reduction from $\RIT$ to bivariate $\NCRANK$} \label{sec-rit-ncrank-reduction}

In this section we give an $\NC$-Turing reduction from $\RIT$ to
bivariate $\NCRANK$. That is, we design an $\NC$ algorithm for $\RIT$
assuming we have an oracle for bivariate $\NCRANK$. Hrubes and
Wigderson in \cite{HW15} give a polynomial time reduction from $\RIT$
to $\NCRANK$ problem \cite[Theorem 2.6]{HW15}. Also they show that
for any given rational formula $\Phi$ there is a log-depth rational
formula that is equivalent to $\Phi$ \cite[Proposition 4.1]{HW15}.


Our key contribution here is to use Cohn's embedding theorem to
transform $\RIT$ problem to the \emph{bivariate} case. Then we
parallelize the Hrubes-Wigderson reduction from $\RIT$ to
$\NCRANK$. In fact, if the input rational formula is already
logarithmic depth then the Hrubes-Wigderson reduction from $\RIT$ to
$\NCRANK$ can be implemented in $\NC$. In this section we design an
$\NC$ algorithm for depth reduction of rational formulas (possibly
with division gates) assuming $\NC$ oracle for bivariate $\NCRANK$.\footnote{ It is an interesting problem to devise an
  $\NC$ algorithm for rational formula depth reduction without oracle
  access to singularity test.} Indeed, the construction of an
equivalent log-depth rational formula, as described in \cite{HW15},
does not appear to be directly parallelizable, as its description
crucially requires rational formula identity testing.\footnote{The
  overall proof in \cite{HW15} is based on the Brent's depth-reduction
  of commutative rational formulas \cite{Br74}. In the commutative
  case addressed by Brent, it turns out that the construction
  procedure does not require oracle access to identity testing and he
  obtains an $\NC$ algorithm for obtaining the depth-reduced formula.}

We show that, indeed, the depth-reduction proof in \cite{HW15} can be
parallelized step by step. However, there are some key points where
our algorithmic proof is different. Firstly, to solve the $\RIT$
instance arising in the depth-reduction proof, we need to recursively
depth-reduce the corresponding subformula and then apply Hrubes-Wigderson
reduction from RIT to $\NCRANK$ on the constructed log-depth subformula.
Secondly, we need to handle an important case arising in the proof (namely, the case (2) in the description of the Normal-Form procedure in the proof of the Lemma \ref{lem-rational-formula-depth-reduction} in the Appendix) which was not significant for the existential argument in
\cite{HW15}. In fact to handle this case, we require an argument based
on Brent's commutative formula depth reduction \cite{Br74}.

In the detailed proof of Lemma \ref{lem-rational-formula-depth-reduction} given in the Appendix, we first sketch our $\NC$ algorithm for depth
reduction of rational formula assuming oracle access to bivariate $\NCRANK$. We highlight and elaborate the key
steps where our proof differs from \cite{HW15}. We need to reproduce
some parts of their proof for completeness, for these parts we just
sketch the argument and refer to \cite{HW15} for the details. Using
this depth reduction algorithm we give an $\NC$ Turing reduction from
$\RIT$ to bivariate $\NCRANK$, which is a main
result of this section.

\subsection{Depth reduction for noncommutative formulas with inversion gates}

The following is a consequence of results in \cite{HW15} and \cite{DM15}.

\begin{theorem}[\cite{HW15}, \cite{DM15}] \label{thm-hwdm}

  Let $\Phi$ be a rational formula of size $s$ computing a
  \emph{non-zero} rational function in $\fF$. If the field $\F$ is
  sufficiently large and $k>2s$ then, at a matrix tuple
  $(M_1, M_2, \ldots, M_n)$ chosen uniformly at random from
  $M_{k \times k} (\F)^n$, the matrix $\hat{\Phi}(M_1, \ldots, M_n)$
  is nonsingular with "high" probability.
\end{theorem}

If $\F$ is small then we can pick the random matrices over a suitable
extension field, and by ``high'' probability we mean, say,
$1-2^{-\Omega(s+n)}$.

Let $\Phi_1$ and $\Phi_2$ be correct rational formulas of size at most
$s$ computing rational functions in $\fF$. By Theorem \ref{thm-hwdm} and a union bound argument, for a
random matrix substitution $(M_1, M_2, \ldots, M_n)$ from
$M_{k \times k} (\F)^n$, inputs to all the inversion gates in $\Phi_1$
and $\Phi_2$ simultaneously evaluate to non-singular matrices with
"high" probability. Hence, for $k$ sufficiently large we have
$\mathcal{D}_{k,\Phi_1} \cap \mathcal{D}_{k,\Phi_2} \neq
\emptyset$. It follows that a random matrix tuple is in
$\mathcal{D}_{k,\Phi_1} \cap \mathcal{D}_{k,\Phi_2}$ with high
probability.


By Theorem \ref{thm-hwdm} and the definition of correct rational
formulas (Definition~\ref{correct-def}), it follows that if $\Phi_1$
and $\Phi_2$ are size $s$ correct formulas that are \emph{not}
equivalent then for a random matrix substitution of dimension $k>2s$
both $\Phi_1$ and $\Phi_2$ are defined and they disagree with high
probability. As noted in Section~\ref{prel-sec}, equivalent formulas
need not have the same domain of definition.

Lemma \ref{lem-rational-formula-depth-reduction} is the main technical
result of this section. It describes an $\NC$ algorithm for
depth-reduction of correct formulas assuming an oracle for bivariate
$\NCRANK$. The next lemma is useful for establishing equivalences of
formulas arising in the proof of Lemma
\ref{lem-rational-formula-depth-reduction}. Suppose $\Phi$ is a
rational formula computing the rational function $\hat{\Phi}$.  For a
gate $v$ in formula $\Phi$, $\Phi_v$ denotes the subformula rooted at
$v$. The formula $\Phi_{v\leftarrow z} \in \F \newbrak{X\cup\{z\}}$ is
obtained from $\Phi$ by replacing subformula $\Phi_v$ with fresh
variable $z$.
\begin{lemma}[local surgery]\label{lem-equivalence}
  Let $\Phi$ be a correct rational formula and
  $v$ be a gate in $\Phi$. Suppose $\Psi$ is a correct rational formula equivalent to $\Phi_v$. Let $\Psi' = (A\cdot z + B) \cdot (C \cdot z +D)^{-1}$ is a formula equivalent to $\Phi_{v \leftarrow z}$ such that $A, B, C, D$ are correct formulas which do not depend upon $z$ and $\hat{C} \cdot \hat{\Delta} + \hat{D} \neq 0$ for any formula $\Delta$ such that $\Phi_{v \leftarrow \Delta}$ is correct. Let
  $\Phi'$ denote the rational formula obtained by replacing
  $z$ in $\Psi'$ by $\Psi$. Then $\Phi'$ is correct and it is equivalent to $\Phi$.
\end{lemma}

\begin{proof}
From the definitions of $\Phi_v$ and $\Phi_{v\leftarrow z}$ it follows that $\Phi = \Phi_{v \leftarrow \Phi_v}$. As $\Phi$ is correct, from the properties of formulas $C,D$ as stated in the lemma it follows that $\hat{C}\hat{\Phi_v}+\hat{D} \neq 0$. Which implies $\hat{C}\hat{\Psi}+\hat{D} \neq 0$ as $\Psi \equiv \Phi_v$. This shows that the formula $\Phi' = (A \cdot \Psi + B) \cdot (C \cdot \Psi + D)^{-1}$ is correct. 
Now let $\tau= (M_1, \ldots, M_n)$ is a matrix tuple in $\mathcal{D}_{\Phi} \cap \mathcal{D}_{\Phi'}$, the intersection of domains of definition of $\Phi$ and $\Phi'$. This implies $\tau \in \mathcal{D}_{\Phi_v}$ as $ \mathcal{D}_{\Phi} \subseteq \mathcal{D}_{\Phi_v}$, $\Phi_v$ being subformula of $\Phi$. Similarly, $\tau \in \mathcal{D}_{\Psi}$ as $\mathcal{D}_{\Phi'} \subseteq \mathcal{D}_{\Psi}$, $\Psi$ being a subformula of $\Phi'$. So $\tau \in  \mathcal{D}_{\Phi_v} \cap \mathcal{D}_{\Psi}$. As $\Phi_v \equiv \Psi$, it follows that $\Phi_v(\tau)= \Psi(\tau)$. 
As $\tau \in \mathcal{D}_{\Phi}$. It implies that $(M_1, M_2,\ldots,M_n, \Phi_v(\tau))= (\tau, \Phi_v(\tau)) \in \mathcal{D}_{\Phi_{v\leftarrow z}}$. 
Similarly, as $\tau \in \mathcal{D}_{\Phi'}$, it follows that $(\tau, \Psi(\tau)) \in \mathcal{D}_{\Psi'}$. As $\Phi_v(\tau)= \Psi(\tau)$, it implies \[(\tau, \Phi_v(\tau))=(\tau, \Psi(\tau)) \in \mathcal{D}_{\Phi_{v\leftarrow z}} \cap \mathcal{D}_{\Psi'}\]
As $\Phi_{v\leftarrow z} \equiv \Psi'$, this implies $\Phi_{v\leftarrow z}(\tau, \Phi_v(\tau)) = \Psi'(\tau, \Psi(\tau))$. But $\Phi_{v\leftarrow z}(\tau, \Phi_v(\tau)) = \Phi(\tau)$ and $\Psi'(\tau, \Psi(\tau))= \Phi'(\tau)$. So we get $\Phi(\tau) = \Phi'(\tau)$ for any $\tau \in \mathcal{D}_{\Phi} \cap \mathcal{D}_{\Phi'}$. Thus proving $\Phi \equiv \Phi'$. 
\end{proof}

\begin{lemma}\label{lem-rational-formula-depth-reduction}
Given a \emph{correct} formula $\Phi$ of size $s$ computing a rational function $f \in \fF$, for sufficiently large $s$ and absolute constants $c,b$
\begin{enumerate}
\item we give an $\NC$ algorithm, with oracle access to bivariate $\NCRANK$, that outputs a correct formula $\Phi'$ of depth at most $c \log _2 s$ which is equivalent to $\Phi$.
\item If a variable $z$ occurs in $\Phi$ at most once then we give an $\NC$ algorithm, with bivariate $\NCRANK$ as oracle, that constructs \emph{correct} rational formulas $A, B, C, D$ which do not depend on $z$ with depth at most $c \log_2 s +b$ and the formula $\Phi' = (A \cdot z+B)\cdot (C \cdot z + D)^{-1}$ is equivalent to $\Phi$. Moreover, the rational function $\hat{C} \hat{\Psi} +\hat{D}\neq 0$ for any formula $\Psi$ such that $\Phi_{z \leftarrow \Psi}$ is correct.
\end{enumerate}
\end{lemma}

The proof of the above lemma is given in the appendix.

\paragraph{Hrubes-Wigderson reduction from $\RIT$ to $\NCRANK$}

Now we recall the polynomial-time reduction from $\RIT$ to $\NCRANK$
from \cite[Theorem 2.6]{HW15}. Given a rational formula $\Phi$ their reduction
outputs an invertible linear matrix $M$ in the variables
$X$.\footnote{Notice that the entries of $M^{-1}$ are elements of the
  skew field $\fF$} Their reduction ensures that the top right entry
of $M^{-1}$ is $\hat{\Phi}$, the rational function computed by the
formula $\Phi$. It turns out that if $\Phi$ is already of logarithmic
depth then their reduction can be implemented in $\NC$.




\begin{theorem}[\cite{HW15}]\label{thm-formula-inverse-completeness}
Let $\Phi$ be a rational formula of size $s$ and depth $O(\log s)$ computing a rational expression in $\fF$ there is an $\NC$ algorithm to construct an invertible linear matrix $M_{\Phi}$ such that the top right entry of $M_{\Phi}^{-1}$ is $\hat{\Phi}$. 
\end{theorem} 

\begin{proof}
  We only briefly sketch the $\NC$ algorithm. Their reduction
  recursively constructs the matrix $M_{\Phi}$, using the formula
  structure of $\Phi$.

Given a formula $\Phi$ we can compute the sizes of all its subformulas
in $\NC$ using a standard pointer doubling algorithm. This allows us
to estimate the dimensions of matrices $M_{\Phi_v}$ for subformulas
$\Phi_v$ for each gate $v$ of $\Phi$. We can also compute in $\NC$ the
precise location for placement of the sub-matrices $M_{\Phi_v}$ inside
the matrix $M_\Phi$ following their construction. Assuming that $\Phi$
is already of logarithmic depth, there are only $O(\log s)$ nested
recursive calls for this recursive procedure. This ensures that the
overall process can be implemented in $\NC$.
\end{proof}

After constructing linear matrix $M_{\Phi}$ such that the top right
entry of $M_{\Phi}^{-1}$ is $\hat{\Phi}$, define matrix $M'$ as $M'= \left(
\begin{array}{cc}
v^T & M_{\Phi} \\
0 & -u
\end{array}
\right)$

where $u,v$ are $1\times k$ vectors, such that $u=(1, 0, \ldots, 0)$
and $v=(0, 0, \ldots, 0, 1)$ where $k$ is the dimension of the matrix
$M_{\Phi}$. It follows that $\hat{\Phi}\neq 0$ iff $M'$ is invertible
in the skew field $\fF$ (see e.g. \cite[Proposition 3.29]{GGOW15}). So
we have the following theorem.


\begin{theorem}[\cite{HW15}]\label{thm-rit-ncrank-reduction}
Let $\Phi$ be a rational formula of size $s$ and depth $O(\log s)$ computing a rational expression in $\fF$ then there is an $\NC$ algorithm to construct a linear matrix $M$ such that $\hat{\Phi} \neq 0$ iff $M$ is invertible in the skew field $\fF$. 
\end{theorem} 
Now, from Lemma \ref{lem-rational-formula-depth-reduction}, Theorem
\ref{thm-rit-ncrank-reduction}, and Theorem \ref{thm-n-rit-2-rit}, we
obtain an $\NC$ Turing reduction from multivariate $\RIT$ to bivariate
$\NCRANK$.

\begin{theorem}
There is a deterministic $\NC$ Turing reduction from $\RIT$ problem to $\NCRANK$ problem for bivariate linear matrices.
\end{theorem}






\paragraph{Concluding Remarks.}

\vspace{2mm}

Motivated by the question whether $\RIT$ and $\NCRANK$ have
deterministic $\NC$ algorithms, we show that multivariate $\RIT$ is
$\NC$-reducible to bivariate $\RIT$ and multivariate $\NCRANK$ is
$\NC$-reducible to bivariate $\NCRANK$.  $\RIT$ is known to be
polynomial-time reducible to $\NCRANK$, and indeed that is how the
polynomial-time algorithm for $\RIT$ works, by reducing to $\NCRANK$
and solving $\NCRANK$. We show that $\RIT$ is deterministic
$\NC$-Turing reducible to $\NCRANK$. We prove this by showing that
noncommutative rational formula depth reduction is $\NC$-Turing
reducible to $\NCRANK$. The main open problem is to obtain deterministic $\NC$ algorithms for bivariate $\NCRANK$ and bivariate
$\RIT$. We also leave open finding an unconditional $\NC$ algorithm
for depth-reduction of noncommutative rational formulas.

\newpage

\bibliographystyle{plain}

\bibliography{references}

\begin{thebibliography}{10}

\bibitem{Am66}
S.~A. Amitsur.
\newblock Rational identities and applications to algebra and geometry.
\newblock {\em Journal of Algebra}, 3:304--359, 1966.

\bibitem{ACM24a}
V.~Arvind, Abhranil Chatterjee, and Partha Mukhopadhyay.
\newblock Black-box identity testing of noncommutative rational formulas in
  deterministic quasipolynomial time.
\newblock {\em CoRR}, abs/2309.1564, 2023. In STOC 2024, to appear.

\bibitem{ACM24b}
V.~Arvind, Abhranil Chatterjee, and Partha Mukhopadhyay.
\newblock Trading determinism for noncommutativity in {E}dmonds' problem.
\newblock {\em CoRR}, abs/2404.07986, 2024.

\bibitem{AJMR19}
Vikraman Arvind, Pushkar~S. Joglekar, Partha Mukhopadhyay, and S.~Raja.
\newblock Randomized polynomial-time identity testing for noncommutative
  circuits.
\newblock {\em Theory Comput.}, 15:1--36, 2019.

\bibitem{Be70}
G.~M. Bergman.
\newblock Skew fields of noncommutative rational functions, after {A}mitsur.
\newblock {\em Sé Schü–Lentin–Nivat, Paris}, 1970.

\bibitem{BJP}
Markus Bl{\"{a}}ser, Gorav Jindal, and Anurag Pandey.
\newblock A deterministic {PTAS} for the commutative rank of matrix spaces.
\newblock {\em Theory Comput.}, 14(1):1--21, 2018.

\bibitem{BW05}
Andrej Bogdanov and Hoeteck Wee.
\newblock More on noncommutative polynomial identity testing.
\newblock In {\em 20th Annual {IEEE} Conference on Computational Complexity
  {(CCC} 2005), 11-15 June 2005, San Jose, CA, {USA}}, pages 92--99, 2005.

\bibitem{Br74}
Richard~P. Brent.
\newblock The parallel evaluation of general arithmetic expressions.
\newblock {\em J. {ACM}}, 21(2):201--206, 1974.

\bibitem{CM23}
Abhranil Chatterjee and Partha Mukhopadhyay.
\newblock The noncommutative edmonds' problem re-visited.
\newblock {\em CoRR}, abs/2305.09984, 2023.

\bibitem{Co71}
P.~M. Cohn.
\newblock The embedding of fir in skew fields.
\newblock {\em Proceedings of the London Mathematical Society}, 23:193--213,
  1971.

\bibitem{Co72}
P.~M. Cohn.
\newblock Universal skew fields of fractions.
\newblock {\em Sympos. Math.}, 8:135--148, 1972.

\bibitem{Cohnrelations}
P.~M. Cohn.
\newblock {\em Free Rings and their Relations}.
\newblock London Mathematical Society Monographs. Academic Press, 1985.

\bibitem{Co90}
P.~M. Cohn.
\newblock An embedding theorem for free associative algebras.
\newblock {\em Mathematica Pannonica}, 1(1):49--56, 1990.

\bibitem{Cohnring}
P.~M. Cohn.
\newblock {\em Introduction to Ring Theory}.
\newblock Springer, 2000.

\bibitem{Cohnfir}
P.~M. Cohn.
\newblock {\em Free Ideal Rings and Localization in General Rings}.
\newblock New Mathematical Monographs. Cambridge University Press, 2006.

\bibitem{DL78}
Richard~A. DeMillo and Richard~J. Lipton.
\newblock A probabilistic remark on algebraic program testing.
\newblock {\em Inf. Process. Lett.}, 7(4):193--195, 1978.

\bibitem{DM15}
Harm Derksen and Visu Makam.
\newblock Polynomial degree bounds for matrix semi-invariants.
\newblock {\em CoRR}, abs/1512.03393, 2015.

\bibitem{GGOW15}
Ankit Garg, Leonid Gurvits, Rafael~Mendes de~Oliveira, and Avi Wigderson.
\newblock A deterministic polynomial time algorithm for non-commutative
  rational identity testing.
\newblock {\em CoRR}, abs/1511.03730, 2015.

\bibitem{GGOW20}
Ankit Garg, Leonid Gurvits, Rafael~Mendes de~Oliveira, and Avi Wigderson.
\newblock Operator scaling: Theory and applications.
\newblock {\em Found. Comput. Math.}, 20(2):223--290, 2020.

\bibitem{Hig}
Graham Higman.
\newblock The units of group-rings.
\newblock {\em Proceedings of the London Mathematical Society},
  s2-46(1):231--248, 1940.

\bibitem{HW15}
Pavel Hrubes and Avi Wigderson.
\newblock Non-commutative arithmetic circuits with division.
\newblock {\em Theory Comput.}, 11:357--393, 2015.

\bibitem{IQS17}
G{\'{a}}bor Ivanyos, Youming Qiao, and K.~V. Subrahmanyam.
\newblock Non-commutative edmonds' problem and matrix semi-invariants.
\newblock {\em Comput. Complex.}, 26(3):717--763, 2017.

\bibitem{IQS18}
G{\'{a}}bor Ivanyos, Youming Qiao, and K.~V. Subrahmanyam.
\newblock Constructive non-commutative rank computation is in deterministic
  polynomial time.
\newblock {\em Comput. Complex.}, 27(4):561--593, 2018.

\bibitem{Ni91}
Noam Nisan.
\newblock Lower bounds for non-commutative computation (extended abstract).
\newblock In {\em Proceedings of the 23rd Annual {ACM} Symposium on Theory of
  Computing, May 5-8, 1991, New Orleans, Louisiana, {USA}}, pages 410--418.
  {ACM}, 1991.

\bibitem{RR07}
Sanguthevar Rajasekaran and John~H. Reif, editors.
\newblock {\em Handbook of Parallel Computing - Models, Algorithms and
  Applications}.
\newblock Chapman and Hall/CRC, 2007.

\bibitem{RS05}
Ran Raz and Amir Shpilka.
\newblock Deterministic polynomial identity testing in non-commutative models.
\newblock {\em Comput. Complex.}, 14(1):1--19, 2005.

\bibitem{Ro80}
L.~H. Rowen.
\newblock {\em Polynomial identities in ring theory, Pure and Applied
  Mathematics}.
\newblock Academic Press Inc., Harcourt Brace Jovanovich Publishers, New York,
  1980.

\bibitem{Sch80}
Jacob~T. Schwartz.
\newblock Fast probabilistic algorithms for verification of polynomial
  identities.
\newblock {\em J. {ACM}}, 27(4):701--717, 1980.

\bibitem{Wy79}
James~C. Wyllie.
\newblock {\em Complexity of Parallel Computation, PhD Thesis}.
\newblock Cornell University, 1979.

\bibitem{Zippel79}
Richard Zippel.
\newblock Probabilistic algorithms for sparse polynomials.
\newblock In Edward~W. Ng, editor, {\em Symbolic and Algebraic Computation,
  {EUROSAM} '79, An International Symposiumon Symbolic and Algebraic
  Computation, Marseille, France, June 1979, Proceedings}, volume~72 of {\em
  Lecture Notes in Computer Science}, pages 216--226. Springer, 1979.

\end{thebibliography}

\newpage

\begin{center}
  \LARGE{Appendix}
\end{center}

\section{Cohn's Embedding Theorem}\label{cohn-sec}

Let $X=\{x_1, x_2,\ldots,x_n\}$ be a set of $n$ noncommuting variables, and let
$x, y$ be a pair of noncommuting variables. The goal of Cohn's
construction \cite{Co90} is to obtain an honest embedding from
$\FX\to \F\angle{x,y}$. Indeed, his construction gives an honest
embedding map even for a countable set of variables
$X=\{x_1,x_2,\ldots\}$.

The first point is that the free noncommutative rings $\FX$ and
$\F\angle{x,y}$ have both enough structure\footnote{Technically, both these rings are semifir \cite{Co90,Cohnfir}.} that guarantees the
following.

\begin{lemma}
  If a homomorphic embedding $\phi:\FX\to \F\angle{x,y}$ can be
  extended to a homomorphic embedding $\phi: \fF\to \fFxy$ then,
  in fact, $\phi$ is an honest embedding.
\end{lemma}  

The reason for this is basically, that if an $n\times n$ $M$ over
$\FX$ is a full matrix then it is invertible with an inverse $M^{-1}$
over the skew field $\fF$. Thus, it suffices to find a homomorphic
embeddings that extends to one between the corresponding free
skew fields.

Cohn solved this problem of finding such an embedding \cite[Theorem
4.7.3]{Cohnfir} by an ingenious construction using skew polynomial
rings. 

\paragraph{Skew Polynomial Rings}

We recall the definition of skew polynomial rings and state some basic
properties (details can be found in Cohn's text \cite[Chapter
  1.1]{Cohnring}).  Let $R$ be an \emph{integral
domain}\footnote{That means $R$ is a, possibly noncommutative, ring
with unity $1$ and without zero divisors.} and let $\sigma:R\to R$ be
a ring \emph{endomorphism}.  Let
\[
A=\{x^n a_n + x^{n-1}a_{n-1}+ \ldots + a_0\mid \textrm{ each } a_i \in R\}
\]
be the set of all formal univariate polynomials in the indeterminate
$x$ which is assumed to not commute with elements of $R$. Addition
of elements in $A$ can be defined component-wise as usual. The multiplication
operation is defined with a "twist" to it, using the ring endomorphism 
$\sigma$, which we briefly explain below.

A \emph{$\sigma$-derivation} on $R$ is defined as an additive
homomorphism $\delta: R\to R$ such that
\[
\delta(ab) = \sigma(a)\delta(b) + \delta(a) b \textrm{ for all } a,b\in R.
\]

We define
\[
ax=x\sigma(a) + \delta(a), \textrm{ for all } a\in R,
\]
which extends to ring multiplication in $A$. Under these operations the set $A$ is the
\emph{skew polynomial ring} denoted $R[x;\sigma,\delta]$. If the $R$-endomorphism 
$\sigma$ is the identity map $1$ and $\delta=0$ in this definition, then we 
obtain the univariate polynomial ring $R[x]$ in which the variable $x$ commutes 
with elements of $R$. 

Following Cohn's construction in his embedding theorem, we consider
skew polynomial rings of the form $R[x; 1, \delta]$, where the endomorphism
$\sigma=1$. We refer to $\delta$ as a derivation and we have:

  \begin{eqnarray}
    \delta(ab) & =& a\delta(b) + \delta(a) b \textrm{ for all } a,b\in
    R, \textrm{ and}\label{skew-eqn1}\\ ax & = & xa + \delta(a),
    \textrm{ for all } a\in R.\label{skew-eqn2}
   \end{eqnarray}

\paragraph{Cohn's Construction}

We now describe the construction, adding some details to the rather terse
description in \cite{Cohnfir}.


We set the integral domain $R$ to be $\FX$ for
$X=\{x_1,x_2,\ldots,\}$. Consider the map $\delta: X \mapsto X$
defined as
\[
\delta(x_i)= x_{i+1}.
\]
The map $\delta$ naturally extends to a unique derivation on $\FX$ as
follows. For scalars $a\in \F$, we define
$\delta (a x_i)=a\delta(x_i)$, by linearity. Next, define $\delta$ on
all monomials in $X^*$. Let
$\delta(x_ix_j)=\delta(x_i)x_j + x_i\delta(x_j)$. In general, for a
degree-$\ell$ monomial $m = x_{i_1} x_{i_2} \ldots x_{i_\ell}$ we
define
\[
\delta(m) =\delta(x_{i_1}x_{i_2}\ldots x_{i_k}) x_{i_{k+1}} x_{i_{k+2}}\ldots
x_{i_\ell} + x_{i_1}x_{i_2}\ldots x_{i_k} \delta ( x_{i_{k+1}} x_{i_{k+2}}\ldots x_{i_{\ell}}).
\]

It is easy to verify that the above definition of $\delta(m)$ is
independent of $k\in[\ell]$. We now extend this definition to the
entire ring $\FX$ by linearity.  By an easy induction on the degree of
polynomials in $\FX$ we obtain the following.


\begin{lemma}
The mapping $\delta$ defined above is a derivation on the ring $R=\FX$.
\end{lemma}

Let $H=R[x;1,\delta]$ be the skew polynomial ring defined by the
derivation $\delta$ described above. By definition, $\delta$ satisfies
Equations~\ref{skew-eqn1} and \ref{skew-eqn2}. Therefore, putting
$a=x_i$ in Equation \ref{skew-eqn2}, for each $i\ge 1$ we have
\[
x_{i+1}=x_ix - xx_i=[x_i,x], 
\]
This actually gives a homomorphic embedding from the ring $\FX$ to
the bivariate ring $\F\angle{x,y}$. To see this, we define a map 
$\beta: X \mapsto \F\angle{x,y}$ as follows:
\begin{itemize}
\item Let $\beta(x_1)=y$. For $i\geq 2$, let $\beta(x_i)=
  [\beta(x_{i-1}),x]$.
\item We can then naturally extend $\beta$ to a homomorphism from $\FX$ to
  $\F\angle{x,y}$, and it is easy to check that it is injective.
\end{itemize}  
Hence we have

\begin{theorem}{\rm\cite[Theorem 4.5.3]{Cohnfir}}\label{beta-homomorph1}
  $\beta$ is a homomorphic embedding from $\FX$ to $\F\angle{x,y}$.
\end{theorem}

%
Furthermore, by the definition of $\delta$, the elements of the skew
polynomial ring $R[x;1,\delta]$ are also polynomials in the ring
$\F\angle{x,y}$.  Indeed, the map $\beta$ can be extended to an
isomorphism from the ring $R[x;1,\delta]$ to $\F\angle{x,y}$ as
follows: for
$f=x^n a_n + x^{n-1}a_{n-1}+ \ldots + a_0 \in R[x;1,\delta]$ define
$\beta(f)= x^n \beta(a_n)+ x^{n-1}\beta(a_{n-1})+ \ldots +
x\beta(a_1)+ \beta(a_0)$.

\begin{theorem}{\rm\cite[Theorem 4.5.3]{Cohnfir}}\label{beta-homomorph}
  $\beta$ is a homomorphic embedding from $\FX$ to $\F\angle{x,y}$.
  Furthermore, $\beta$ is an isomorphism from $R[x;1,\delta]$ to
  $\F\angle{x,y}$.
  \end{theorem}

Using properties of the field of fractions of the skew polynomial
ring $R[x;1,\delta]$ Cohn shows that $\beta$ extends to an
embedding between the skew fields.

\begin{theorem}[Cohn's embedding theorem]
  The embedding map $\beta:\FX\to \F\angle{x,y}$ extends to an
  embedding $\beta : \fF\to \fFxy$ which implies that $\beta$
  is an honest embedding.
\end{theorem}  

\paragraph{Proof of Lemma~\ref{lem-simple-formula-depth-reduction}}

\vspace{2mm}

\begin{proof}
  First we describe a recursive construction to compute a formula
  $\Phi'$ equivalent to $\Phi$ and inductively prove that the depth of
  $\Phi'$ is $c \log_2 s$ for an absolute constant $c$. Then we
  analyze the parallel time complexity of the construction and prove that
  it can be implemented in $\NC$.

  Let $A_\Phi$ be $s \times s$ matrix such that for gates
  $u,v \in \Phi$, $(u,v)^{th}$ entry of $A_\Phi$ is $1$ if gate $v$ is
  a descendent of gate $u$. Using the well-known pointer doubling
  strategy (see e.g. \cite{Wy79}, \cite{RR07}) we can compute matrix
  $A_\Phi$ in $\NC$. So by adding elements in each row of $A_\Phi$, we
  can compute $wt(u)$ (that is the number of descendants of gate
  $u \in \Phi$) in $\NC$. Let $v$ be a gate in $\Phi$ such that
  $\frac{s}{3} \leq wt(v) < \frac{2s}{3}$. Such a gate always exists
  by a standard argument.
Since we can compute the number of descendants of a gate in $\NC$, we
can also find such a gate $v$ in $\NC$, by simply having a processor
associated to each gate to check the above inequalities. Now we are
ready to describe recursive construction of $\Phi'$.
\begin{enumerate}
\item In $\NC$ find a gate $v$ in $\Phi$ such that $\frac{s}{3} \leq wt(v) < \frac{2s}{3}$.
\item Let $r=v_0$ be the root of $\Phi$ and
  $v_1,v_2, \ldots, v_{\ell-1}$ be gates on $r$ to $v$ path in
  $\Phi$. Let $v=v_\ell$. For $1\leq i \leq \ell$, let $u_i$ denote a
  sibling of $v_i$. Let $S_1$ be collection of all indices $j$ such
  that $1\leq j \leq \ell$, $v_j$ is a product gate and is a right
  child of its parent. Similarly let $S_2$ be collection of all
  indices $j$ such that $1\leq j \leq \ell$, $v_j$ is a product gate
  and is a left child of its parent. Define formula
  $\Psi_1 = \prod _{j \in S_1} \Phi_{u_j}$. The product is computed
  using sequence of multiplication gates, starting with $\Phi_{u_j}$
  for the first $u_j$ (one with smallest index $j\in S_1$) each
  multiplication gate multiplies the product so far from right by
  $\Phi_{u_j}$ for the next gate $u_j$, $j\in S_1$, along the root to
  $v$ path. Similarly define formula
  $\Psi_2 = \prod _{j \in S_2} \Phi_{u_j}$. Let $\Psi_3$ be a formula
  obtained from $\Phi$ by replacing subformula $\Phi_v$ by zero.
\item Recursively in parallel compute log-depth formulas
  $\Psi_1', \Psi_2', \Psi_3', \Phi_v'$ equivalent to $\Psi_1$,
  $\Psi_2$, $\Psi_3$ and $\Phi_v$ respectively.
\item Define formula
  $\Phi'$ as $(\Psi_1'\cdot \Phi_v')\cdot \Psi_2'+\Psi_3$.
\end{enumerate} 
From the definitions of $\Psi_1, \Psi_2$ and $\Psi_3$ it is clear that
the polynomial computed by $\Phi$ equals
$(\hat{\Psi_1}\cdot \hat{\Phi_v})\cdot \hat{\Psi_2}+\hat{\Psi_3}$,
where $\hat{\Psi_1},\hat{\Psi_1},\hat{\Psi_1}, \hat{\Phi_v}$ are the
polynomials computed by $\Psi_1$, $\Psi_2$, $\Psi_3$ and $\Phi_v$
respectively. Hence, $\Phi'$, defined in Step 4, is equivalent to $\Phi$.

Let $d(s)$ denote the upper bound on the depth of the formula output
by the above procedure if size $s$ formula is given to it as input. We
use induction on the size $s$ to prove that $d(s)\leq c \log_2 s$. As
$\Psi_1, \Psi_2$ are disjoint subformulas of $\Phi_{v\leftarrow z}$,
clearly we have $|\Psi_1|+|\Psi_2| \leq |\Phi_{v\leftarrow z}|$. Since
$|\Phi_v| \geq \frac{s}{3}$, it implies
$|\Psi_1|, |\Psi_2| \leq |\Phi_{v\leftarrow z}|\leq
\frac{2s}{3}$. From the definition of $\Psi_3$, it is clear that
$|\Psi_3| \leq |\Phi_{v\leftarrow z}|\leq \frac{2s}{3}$. So the size
of each formula $\Psi_1$, $\Psi_2$, $\Psi_3$ and $\Phi_v$ is upper
bounded by $\frac{2s}{3}$. Hence, inductively, for each of the
formulas $\Psi_1', \Psi_2', \Psi_3', \Phi_v'$ the depth is upper
bounded by $c \log_2 \frac{2s}{3}$. As $\Phi'$ is obtained from
$\Psi_1', \Psi_2', \Psi_3', \Phi_v'$ using two multiplications and an
addition as in Step 4, it follows that the depth of
$\Phi'=d(s)\leq c\log_2\frac{2s}{3} + 3$. Choosing
$c\geq \frac{3}{(\log_2 3-1)}$ we get
$d(s)\leq c\log_2\frac{2s}{3} + 3 \leq c \log _2 s$. This completes
the induction, proving that the depth of $\Phi'$ is at most
$c \log_2 s$.

Let $t(s)$ denotes parallel time complexity of the above
procedure. Since Steps 1, 2, 4 can be implemented in $\NC$ they
together take $(\log s)^k$ parallel time for an absolute constant
$k$. As all the recursive calls in Step 3 are processed in parallel,
we have the recurrence $t(s)\leq t(2s/3)+(\log s)^k$. Hence,
$t(s)\leq (\log s)^{(k+1)}$. This shows that the above procedure can
be implemented in $\NC$, completing the proof of the theorem.
\end{proof}

\paragraph{Proof of Theorem~\ref{thm-parallel-effective-Higman}}

\vspace{2mm}

\begin{proof}
  By Lemma \ref{lem-simple-formula-depth-reduction} we have an $\NC$
  algorithm to convert every entry of $A$ to a log-depth formula. We
  will first describe parallel algorithm for Higman linearization of
  single polynomial $f$ given by noncommutative log-depth formula
  $\Phi$. Higman linearization of a polynomial matrix can be handled
  similarly.

\begin{claim} \label{clm-higman-linearization-for-polynomial}
Given a noncommutative formula $\Phi$ of size $s$ and depth $O(\log s)$ computing a polynomial $f$ in $\FX$. We can compute Higman linearization of $f$ in deterministic $\NC$. More precisely, we can compute a linear matrix $L_\Phi \in \FX^{(s+1)\times (s+1)}$, invertible upper and lower triangular matrices $P$,$Q$ $\in \FX^{(s+1)\times (s+1)}$ with all diagonal entries $1$ such that $P(f \oplus I_s)Q = L_{\Phi}$. 
All the entries of $P,Q$ are computable by algebraic branching programs of size $\poly(s)$. Moreover, $f \not\equiv 0$ iff $L_{\Phi}$ is a full noncommutative rank matrix.
\end{claim}
\begin{proof}
  The basic idea is to compute the Higman linearization recursively in
  parallel. Since the formula size is known and for every $+$ or
  $\times$ gate in $\Phi$ we need to add a new row and column for the
  Higman process, we know exactly the dimension of the final matrix
  $L$ and the precise location for placement of the submatrices inside
  $L$, corresponding to each subformula of $\Phi$. To obtain Higman
  linearization of $\Phi$ we recursively in parallel compute Higman
  linearization of the subformulas rooted at the left and the right
  child of the root. Based on whether at root of $\Phi$ we have a
  $+$-gate or a $\times$-gate, we appropriately compose Higman
  linearizations of the sub-formulas rooted at the left and the right
  child of the
  root. 
  As depth of $\Phi$ is logarithmic, there are only $O(\log s)$ nested
  recursive calls, this ensures that the overall process can be
  implemented in $\NC$.

We now give the details of the inductive proof. 

Let $r$ be the root of $\Phi$, and $g$ and $h$ be the polynomials
computed at the left and right child of $r$, respectively. If $r$ is a
$\times$-gate, the matrix obtained after the first step of the Higman
Linearization is

\[
\left(
\begin{array}{cc}
1 & g \\
0 & 1
\end{array}
\right)
\left(
\begin{array}{cc}
g\cdot h & 0 \\
0 & 1
\end{array}
\right)
\left(
\begin{array}{cc}
1 & 0 \\
-h & 1
\end{array}
\right) ~=~
\left(
\begin{array}{cc}
0 & g \\
-h & 1
\end{array}
\right)
\]

If $r$ is a $+$-gate, we do not explicitly deal with it, as we are
eventually interested only in \emph{linearizing} the
matrix. Nevertheless, for the $+$-gate case we carry out following
step.\footnote{This facilitates the presentation and calculation of
  parallel placement of linear submatrices corresponding to the
  recursive calls inside the final Higman linearized matrix for $f$.}

\[
\left(
\begin{array}{cc}
1 & g \\
0 & 1
\end{array}
\right)
\left(
\begin{array}{cc}
g+h & 0 \\
0 & 1
\end{array}
\right)
\left(
\begin{array}{cc}
1 & 0 \\
-1 & 1
\end{array}
\right) ~=~
\left(
\begin{array}{cc}
h & g \\
-1 & 1
\end{array}
\right)
\]
It reduces computing Higman linearization of $g+h$ to computing it for
$g$ and for $h$.

We can in $\NC$ compute the size of sub-tree rooted at any gate of
$\Phi$. Let $s_1$ and $s_2$ be sizes of the left and right sub-trees
of the root $r$. Let $w=-h$ if $r$ is a $\times$-gate and is equal to
$h$ when $r$ is a $+$-gate. Suppose we compute Higman linearization of
$g$ and $w$ recursively in parallel. More precisely, we obtain
invertible upper and lower triangular matrices
$P_g,Q_g \in \FX^{(s_1+1)\times (s_1+1)}$ respectively, with all
diagonal entries $1$ and linear matrix
$L_g \in \FX^{(s_1+1)\times (s_1+1)}$ such that
$P_g~ (g \oplus I_{s_1})~ Q_g=L_g$. The entries of $P_g, Q_g$ are
given by ABPs of size $\poly(s_1)$. Similarly, we obtain invertible
upper and lower triangular matrices
$P_{w},Q_{w} \in \FX^{(s_2+1)\times (s_2+1)}$ respectively, with all
diagonal entries $1$ and linear matrix
$L_w \in \FX^{(s_2+1)\times (s_2+1)}$ such that
$P_w~ (w \oplus I_{s_2})~ Q_w = L_w$. The entries of $P_w,Q_w$ are
given by ABPs of size $\poly(s_2)$. Let $R_g \in \FX^{1\times s_1}$ be
a row matrix obtained by dropping $(1,1)^{th}$ entry from the first
row of $P_g$. Similarly, column matrix $C_g$ is obtained by
dropping$(1,1)^{th}$ entry from the first column of
$Q_g$. Analogously, define row and column matrices $R_w$ and
$C_w$. Let $P_g',Q_g'$ denote bottom right $s_1 \times s_1$ blocks of
matrices $P_g$ and $Q_g$ respectively. Similarly, let $P_w',Q_w'$
denote bottom right $s_2 \times s_2$ blocks of matrices $P_w$ and
$Q_w$ respectively. So we have,
\[
L_g = \left(
\begin{array}{cc}
1 & R_g \\
0 & P_g'
\end{array}
\right)
\left(
\begin{array}{cc}
g & 0 \\
0 & I_{s_1}
\end{array}
\right)
\left(
\begin{array}{cc}
1 & 0 \\
C_g & Q_g'
\end{array}
\right)
=
\left(
\begin{array}{cc}
g+R_gC_g & R_gQ_g' \\
P_g'C_g & P_g'Q_g'
\end{array}
\right)
\]   
Similarly, we have
\[
L_w = \left(
\begin{array}{cc}
1 & R_w \\
0 & P_w'
\end{array}
\right)
\left(
\begin{array}{cc}
w & 0 \\
0 & I_{s_2}
\end{array}
\right)
\left(
\begin{array}{cc}
1 & 0 \\
C_w & Q_w'
\end{array}
\right)
=
\left(
\begin{array}{cc}
w+R_wC_w & R_wQ_w' \\
P_w'C_w & P_w'Q_w'
\end{array}
\right)
\]   

Now we are ready to define matrices $P,Q,L$ using matrices
$P_g,Q_g,P_w,Q_w$ so that $P(f\oplus I_s)Q = L$.

\begin{itemize}
\item {\bf $r$ is $\times$-gate}

After the first step of Higman linearization on $f$ the matrix
obtained is
$\left(
\begin{array}{cc}
0 & g \\
w & 1
\end{array}
\right) $.

When we obtain Higman linearization for $g$, $w$ recursively, these
polynomials sit at $(1,1)^{th}$ entry of $g\oplus I_{s_1}$ and
$w \oplus I_{s_2}$ respectively. Whereas in the matrix above, the
polynomials $g,w$ sit at $(1,2)^{th}, (2,1)^{th}$ entries
respectively. Now we define block matrices $\tilde{P_g},\tilde{P_w}$
and $\tilde{Q_g},\tilde{Q_w}$ for performing row and column
operations, respectively, on the appropriate rows and columns of the
matrix, taking into account the location of $g$ and $h$. Whenever we
are carrying out linearization for $g$ it keeps the block
corresponding to linearization of $h$ intact and vice-versa.  Define
$\tilde{P_g}, \tilde{P_w}$ as
\[
\tilde{P_g}=
\left(
\begin{array}{c|c}
\begin{array}{cc}
1 & 0\\
0 & 1
\end{array} & 
\begin{array}{cc}
R_g & 0\\
0 & 0
\end{array}\\
\hline
0 & 
\begin{array}{cc}
P_g' & 0\\
0 & I_{s_2}
\end{array}
\end{array}
\right), \tilde{P_w}=
\left(
\begin{array}{c|c}
\begin{array}{cc}
1 & 0\\
0 & 1
\end{array} & 
\begin{array}{cc}
0 & 0\\
0 & R_w
\end{array}\\
\hline
0 & 
\begin{array}{cc}
I_{s_1} & 0\\
0 & P_w'
\end{array}
\end{array}
\right)
\]

Similarly, define $\tilde{Q_g}, \tilde{Q_w}$ as
\[
\tilde{Q_g}=
\left(
\begin{array}{c|c}
\begin{array}{cc}
1 & 0\\
0 & 1
\end{array} & 0\\
\hline
\begin{array}{cc}
0 & C_g\\
0 & 0
\end{array} & 
\begin{array}{cc}
Q_g' & 0\\
0 & I_{s_2}
\end{array}
\end{array}
\right), \tilde{Q_w}=
\left(
\begin{array}{c|c}
\begin{array}{cc}
1 & 0\\
0 & 1
\end{array} & 0\\
\hline
\begin{array}{cc}
0 & 0\\
C_w & 0
\end{array} & 
\begin{array}{cc}
I_{s_1} & 0\\
0 & Q_w'
\end{array}
\end{array}
\right)
\]

$\tilde{P_g}, \tilde{Q_g}$ carry out linearization of $g$, $\tilde{P_w}, \tilde{Q_w}$ carry out linearization of $w$. Define $P_1 = \left(
\begin{array}{cc}
1 & g \\
0 & 1
\end{array}
\right) \oplus I_{s_1+s_2}
$ and $Q_1 = \left(
\begin{array}{cc}
1 & 0 \\
w & 1
\end{array}
\right) \oplus I_{s_1+s_2}
$, which would carry out the first step of linearization.

Finally, define $P=\tilde{P_w}\tilde{P_g}P_1$. We can see that $P=\left(
\begin{array}{c|c}
\begin{array}{cc}
1 & g\\
0 & 1
\end{array} & \begin{array}{cc}
R_g & 0\\
0 & R_w
\end{array}\\
\hline
0 & 
\begin{array}{cc}
P_g' & 0\\
0 & P_w'
\end{array}
\end{array}
\right)$. Similarly, define $Q=\tilde{Q_w}\tilde{Q_g}Q_1=\left(
\begin{array}{c|c}
\begin{array}{cc}
1 & 0\\
w & 1
\end{array} & 0\\
\hline
\begin{array}{cc}
0 & C_g\\
C_w & 0
\end{array} & 
\begin{array}{cc}
Q_g' & 0\\
0 & Q_w'
\end{array}
\end{array}
\right)$

As $P_g'$ and $P_w'$ are bottom right blocks of the upper triangular
matrices $P_g$ and $P_w$, it follows that $P_g'$, $P_w'$, and hence
also $P$ is invertible, and also upper triangular with all diagonal
entries $1$. Similarly, $Q$ is invertible and lower triangular with all
diagonal entries $1$. As $P,Q$ are realized as a matrix product as
defined above and entries of $P_g, Q_g$ and $P_w, Q_w$ are given by
ABPs of size $\poly(s_1), \poly(s_2)$ respectively, we can construct
ABPs of $\poly(s)$ size for entries of $P$ and $Q$ in $\NC$.

Now, define $L_{\Phi}=P(f\oplus I_s)Q =\left(
\begin{array}{c|c}
\begin{array}{cc}
0 & g+R_gC_g\\
w+R_wC_w & 1
\end{array} & \begin{array}{cc}
R_gQ_g' & 0\\
0 & R_wQ_w'
\end{array}\\
\hline
\begin{array}{cc}
0 & P_g'C_g\\
P_w'C_w & 0
\end{array} & 
\begin{array}{cc}
P_g'Q_g' & 0\\
0 & P_w'Q_w'
\end{array}
\end{array}
\right)$

Clearly, $L_{\Phi}$ is linear as $L_g$ and $L_w$ are linear. Note that
we do not need to compute $L_{\Phi}$ as a product $P(f\oplus
I_s)Q$. If we explicitly know linear entries of $L_g, L_w$ (which we
do know recursively) we can explicitly compute linear entries of
$L_\Phi$.


\item{\bf $r$ is $+$-gate}

  This case is handled similarly. In case of $+$-gate, the matrix
  obtained after first step is
  $\left(
\begin{array}{cc}
w & g \\
-1 & 1
\end{array}
\right)$.
So to ensure the Higman linearized matrices of $g$ and $w$ are
correctly placed inside the Higman linearized matrix for $f$ we define
$P,Q$ as
\[
P=\left(
\begin{array}{c|c}
\begin{array}{cc}
1 & g\\
0 & 1
\end{array} & \begin{array}{cc}
R_g & R_w\\
0 & 0
\end{array}\\
\hline
0 & 
\begin{array}{cc}
P_g' & 0\\
0 & P_w'
\end{array}
\end{array}
\right), Q=\left(
\begin{array}{c|c}
\begin{array}{cc}
1 & 0\\
-1 & 1
\end{array} & 0\\
\hline
\begin{array}{cc}
0 & C_g\\
C_w & 0
\end{array} & 
\begin{array}{cc}
Q_g' & 0\\
0 & Q_w'
\end{array}
\end{array}
\right)
\]
and define $L_\Phi$ as 
\[L_\Phi=P(f\oplus I_s)Q=\left(
\begin{array}{c|c}
\begin{array}{cc}
w+R_wC_w & g+R_gC_g\\
-1 & 1
\end{array} & \begin{array}{cc}
R_gQ_g' & R_wQ_w'\\
0 & 0
\end{array}\\
\hline
\begin{array}{cc}
0 & P_g'C_g\\
P_w'C_w & 0
\end{array} & 
\begin{array}{cc}
P_g'Q_g' & 0\\
0 & P_w'Q_w'
\end{array}
\end{array}
\right)\] Clearly, any entry of $L_\Phi$ is a scalar or some entry of
$L_g$ or $L_w$, so $L_\Phi$ is linear. Again, entries of $L_\Phi$ can
be computed explicitly, given the matrices $L_g$ and $L_w$
explicitly. This completes the Case 2.

As $P,Q$ are invertible, it implies $\ncrk(f\oplus I_s) = \ncrk(L)$,
which implies $L$ is full iff $f \not\equiv 0$. This proves the
claim.
\end{itemize}
\end{proof}

In the general case, we need to Higman linearize an $n \times n$
polynomial matrix $A$. We will in parallel compute Higman
linearization matrices for each entry of $A$ using the $\NC$ algorithm
described in the above claim. Let $M_{i,j}$ be
$(s_{i,j}+1) \times (s_{i,j}+1)$ linear matrix corresponding to
$(i,j)^{th}$ entry for $1\leq i,j \leq n$. That is
$M_{i,j}= P_{i,j} (A_{i,j}\oplus I_{s_{i,j}}) Q_{i,j}$ where
$P_{i,j}, Q_{i,j}$ are invertible upper and lower triangular matrices
respectively. We have polynomial sized ABPs for entries of $P_{i,j}$
and $Q_{i,j}$.

Let $M_{i,j}'$ denote $s_{i,j}\times s_{i,j}$ bottom right block of $M_{i,j}$. Let $R_{i,j}$ denote $1 \times s_{i,j}$ row matrix obtained by dropping first entry of first row of $M_{i,j}$. Similarly, let $C_{i,j}$ be $s_{i,j} \times 1$ column matrix obtained by dropping first entry of the first column of $M_{i,j}$. 
Then, the matrix $L$ is a $2\times 2$ block matrix such that
\begin{enumerate}
\item The top left block of $L$ is $n\times n$ and for $1\leq i,j \leq n$, $(i,j)^{th}$ entry of the block is $(1,1)^{th}$ entry of $M_{i,j}$.
\item The top right block of $L$ is the matrix 
\[\left(\begin{array}{cccc}
R_{1,1}~R_{1,2}~\ldots R_{1,n} & ~~ & ~~ & ~~\\
~~ & R_{2,1}~R_{2,2}~\ldots R_{2,n}& ~~ & ~~\\
~~ & ~~ & \ddots&~~\\
~~&~~&~~&R_{n,1}~R_{n,2}~\ldots R_{n,n}
\end{array} \right)\]
\item The bottom left block of $L$ is the matrix $[B_1 B_2 \ldots B_n]^T$ where for $1\leq i\leq n$, $B_i$ is the matrix \[\left(\begin{array}{cccc}
C_{i,1} & ~~ & ~~ & ~~\\
~~ & C_{i,2}& ~~ & ~~\\
~~ & ~~ & \ddots&~~\\
~~&~~&~~&C_{i,n}
\end{array} \right)\]
\item The bottom right block of $L$ is a matrix $\left(\begin{array}{cccc}
D_1 & ~~ & ~~ & ~~\\
~~ & D_2& ~~ & ~~\\
~~ & ~~ & \ddots&~~\\
~~&~~&~~&D_n
\end{array} \right)$ where $D_i$ is a matrix $\left(\begin{array}{cccc}
M_{i,1}' & ~~ & ~~ & ~~\\
~~ & M_{i,2}'& ~~ & ~~\\
~~ & ~~ & \ddots&~~\\
~~&~~&~~&M_{i,n}'
\end{array} \right)$
\end{enumerate}
   
All unspecified entries in the above matrices are zero. As done in
the proof of Claim \ref{clm-higman-linearization-for-polynomial},
we can easily define invertible upper and lower triangular matrices
$P$ and $Q$ such that $P(A\oplus I_k)Q =L$. Also we can obtain ABPs
computing entries of $P,Q$ from ABPs for the entries of the matrices
$P_{i,j}, Q_{i,j}$, $1\leq i, j \leq n$. Since $P,Q$ are invertible,
it implies that $\ncrk(L)$ is equal to $\ncrk(A\oplus I_k)$
which is equal to $\ncrk(A)+k$. Clearly, the dimension of the bottom
right block of $L$ is $O(n^2 \cdot s)$ which implies similar bound on
$k$. This completes the proof of the theorem.
\end{proof}

\paragraph{Proof of Lemma~\ref{lem-rational-formula-depth-reduction}}

\vspace{2mm}

\begin{proof}
  We give a recursive construction for both the parts and prove the
  correctness by induction on $s$, the size of the formula $\Phi$.

{\bf Depth-reduce($\Phi$)}\\
{\bf Input:} A correct formula $\Phi$ of size $s$ computing a rational function in $\fF$.\\
{\bf Output:} A correct formula $\Phi'$ of depth at most $c \log_2 s$ such that $\Phi \equiv \Phi'$.
\begin{enumerate}
\item Find a gate $v \in \Phi$ such that $\frac{s}{3}< wt(v) \leq \frac{2s}{3}$.
\item In Parallel construct formulas $\Psi$, $\Delta$ such that $\Psi= \text{Depth-Reduce}(\Phi_v)$ and $\Delta = \text{Normal-Form}(\Phi_{v\leftarrow z},z)$.
\item Obtain formula $\Phi'$ from $\Delta$ by replacing $z$ in $\Delta$ by $\Psi$.
\item Output $\Phi'$.
\end{enumerate}

{\bf Normal-Form($\Phi$, $z$)}\\
{\bf Input:} A correct formula $\Phi$ of size $s$ computing a rational function in $\fF$ and a variable $z \in \Phi$ which appears at most once in $\Phi$\\
{\bf Output:} A correct formula $\Phi'$ which is of the form \[
\Phi' = (Az+B)(Cz+D)^{-1}
\]
where $A, B, C$ and $D$ are correct rational formulas which do not depend on $z$ with depth at most $c \log_2 s+b$. Moreover, the rational function $\hat{C} \hat{\Psi} +\hat{D}\neq 0$ for any formula $\Psi$ such that $\Phi_{z \leftarrow \Psi}$ is correct.

Let $v_1, v_2, \ldots, v_\ell=z$ be gates on the path from root $r$ of
$\Phi$ to the leaf gate $z$, such that $v_j$ is not an inverse
gate. So parent of each $v_j$ has two children, and let $u_i$ denote
the sibling of $v_i$. Use pointer doubling based parallel algorithm
(as mentioned in the proof of Lemma
\ref{lem-simple-formula-depth-reduction}) to compute $wt(u_i)$ and
$wt(v_i)$ for all $i\in [\ell]$. We call gate $v_i$ for $i\in [\ell]$
as a \emph{balanced} gate if
$wt(\Phi_{v_i}), wt(\Phi_{v_i \leftarrow z'}) \leq \frac{5s}{6}$. Now
we consider two cases.

\begin{enumerate}
\item There exist a balanced gate $v_i$:\\
\begin{enumerate}
\item Let $v= v_i$. In parallel compute formulas $\Psi_1, \Psi_2$ such that 
\begin{eqnarray*}
\Psi_1&=&(A_1z'+B_1)(C_1z'+D_1)^{-1}=\text{Normal-Form}(\Phi_{v \leftarrow z'},z')\\
\Psi_2&=&(A_2z+B_2)(C_2z+D_2)^{-1}=\text{Normal-Form}(\Phi_{v},z)
\end{eqnarray*}
\item Define formulas $A, B, C, D$ as $A_1\cdot A_2 + B_1 \cdot C_2$, $A_1\cdot B_2 + B_1 \cdot D_2$, $C_1\cdot A_2 + D_1 \cdot C_2$ and $C_1\cdot B_2 + D_1 \cdot D_2$ respectively.
\item Let $\Phi'=(A\cdot z+B)\cdot(C \cdot z+D)^{-1}$
\item Output $\Phi'$ and halt
\end{enumerate}
\item There does not exist a balanced gate:\\
In this case, we can prove that there is a unique $i\in[\ell]$ such that $wt(u_i)> \frac{s}{6}$.
\begin{enumerate}
\item Let $v$ be parent of the gate $v_i$.
\item In Parallel find formulae $\Psi_1,\Psi_2,\Psi_3,\Psi_4$ such that 
\begin{eqnarray*}
\Psi_1&=&(A_1z'+B_1)(C_1z'+D_1)^{-1}=\text{Normal-Form}(\Phi_{v \leftarrow z'},z')\\
\Psi_2&=&(A_2z+B_2)(C_2z+D_2)^{-1}=\text{Normal-Form}(\Phi_{v_i},z)\\
\Psi_3&=&\text{Depth-Reduce}(\Phi_{u_i})\\
\Psi_4&=& (Az+B)(Cz+D)^{-1}=\text{Normal-Form}(\Phi'',z)
\end{eqnarray*}
where $\Phi''$ is obtained by replacing sub-tree rooted at $u_i$ by $0$ in $\Phi$.
\item Using the algorithm of Theorem \ref{thm-rit-ncrank-reduction}
  check if {$\hat{\Psi_3}\equiv 0$} in $\NC$ with oracle access to
  bivariate $\NCRANK$. If $\hat{\Psi_3}\equiv 0$ then output $\Psi_4$
  and halt.
\item If $\hat{\Psi_3}\not\equiv 0$ then let $\Phi'=(A\cdot z+B)\cdot(C \cdot z+D)^{-1}$, where
\begin{equation*}
A = \begin{cases}
A_1A_2+B_1C_2+A_1\hat{\Psi_3}C_2&\text{ if $v_i$ is a $+$-gate}\\
A_1\hat{\Psi_3}A_2+B_1C_2 &\text{if $v_i$ is a $\times$-gate and is a right child of $v$}\\
A_1A_2+B_1\hat{\Psi_3}^{-1}C_2&\text{if $v_i$ is a $\times$-gate and is a left child of $v$} 
\end{cases}
\end{equation*}
\begin{equation*}
B = \begin{cases}
A_1B_2+B_1D_2+A_1\hat{\Psi_3}D_2&\text{ if $v_i$ is a $+$-gate}\\
A_1\hat{\Psi_3}B_2+B_1D_2 &\text{if $v_i$ is a $\times$-gate and is a right child of $v$}\\
A_1B_2+B_1\hat{\Psi_3}^{-1}D_2&\text{if $v_i$ is a $\times$-gate and is a left child of $v$}
\end{cases}
\end{equation*}
\begin{equation*}
C = \begin{cases}
C_1A_2+D_1C_2+C_1\hat{\Psi_3}C_2&\text{ if $v_i$ is a $+$-gate}\\
C_1\hat{\Psi_3}A_2+D_1C_2 &\text{if $v_i$ is a $\times$-gate and is a right child of $v$}\\
C_1A_2+D_1\hat{\Psi_3}^{-1}C_2&\text{if $v_i$ is a $\times$-gate and is a left child of $v$}
\end{cases}
\end{equation*}
\begin{equation*}
D = \begin{cases}
C_1B_2+D_1D_2+C_1\hat{\Psi_3}D_2&\text{ if $v_i$ is a $+$-gate}\\
C_1\hat{\Psi_3}B_2+D_1D_2&\text{if $v_i$ is a $\times$-gate and is a right child of $v$}\\
C_1B_2+D_1\hat{\Psi_3}^{-1}D_2&\text{if $v_i$ is a $\times$-gate and is a left child of $v$}
\end{cases}
\end{equation*}
\item Output $\Phi'$ and halt.
\end{enumerate}
\end{enumerate}

Case 2 above is not explicitly dealt with in \cite{HW15} as their
focus is on the \emph{existence} of a log-depth formula equivalent to
$\Phi$. In contrast to that, in our case we want to
\emph{algorithmically construct} log-depth formula $\Phi'$ equivalent
to $\Phi$. This makes the details of Case 2 crucial as the
construction of $\Phi'$ in Case 2 depends on whether
$\hat{\Phi_{u_i}} \equiv 0$ or not. To solve this $\RIT$ instance we
need to recursively compute log-depth formula $\Psi_3$ equivalent to
$\Phi_{u_i}$ and then invoke algorithm of theorem
\ref{thm-rit-ncrank-reduction} to carry out $\RIT$ test in NC with
oracle access to bivariate $\NCRANK$, as in Step (c).
 
We first prove the correctness of both the algorithms described above using induction on $s$, then we analyze the parallel complexity of both the algorithms. We will choose appropriate constants $c, b$ during the proof.

\paragraph{Correctness of the algorithm Depth-Reduce} 
We know that there exists a gate $v\in \Phi$ such that $\frac{s}{3}< wt(v) \leq \frac{2s}{3}$. As in proof of Lemma \ref{lem-simple-formula-depth-reduction} we can find such a gate $v$ as required by Step 1 of the Depth-Reduce algorithm.  
Clearly, the formulas $\Phi_v$ and $\Phi_{v\leftarrow z}$ are of size at most $2s/3$. Using inductive assumption, we can construct a correct formula $\Psi$ such that depth of $\Psi$ is at most $c \log_2 \frac{2s}{3}$ and $\Psi \equiv \Phi_v$. Again using inductive assumption we can construct correct formulas $A, B, C, D$ (which do not depend on $z$) of depth at most $c \log_2 \frac{2s}{3}+b$ such that the formula $\Delta = (A\cdot z+B)\cdot(C\cdot z+D)^{-1}$ is equivalent to $\Phi_{v \leftarrow z}$.
Since $\Phi$ is equal to the formula obtained from $\Phi_{v\leftarrow z}$ by replacing $z$ by $\Phi_v$ and $\Phi$ is correct so from inductive assumption it follows that $\hat{C} \hat{\Phi_v} + \hat{D} \neq 0$. 

As $\Psi \equiv \Phi_v$, $\Delta \equiv \Phi_{v \leftarrow z}$ and $\hat{C} \hat{\Phi_v} + \hat{D} \neq 0$ from Lemma \ref{lem-equivalence} it follows that $\Phi'$ is correct and $\Phi'\equiv \Phi$. Since $\Phi' = (A\cdot \Psi+B)\cdot(C\cdot \Psi+D)^{-1}$ we get that depth of $\Phi'$ is at most $c \log_2 \frac{2s}{3}+b+4$. By choosing constant $c \geq \frac{b+4}{(\log_2 3 -1)}$, we get that the depth of $\Phi'$ is at most $c \log_2 s$. Completing the inductive argument for the correctness proof of the algorithm Depth-Reduce.  

\paragraph{Correctness of the algorithm Normal-Form}
  In case (1) we know that there exists a balanced gate $v=v_i$. We have $wt(\Phi_{v}), wt(\Phi_{v \leftarrow z'}) \leq \frac{5s}{6}$. So by inductive assumption we know that the formulas $A_j,B_j, C_j, D_j$ for $j\in \{1,2\}$ are correct, and their depths are at most $c \log_2 \frac{5s}{6} + b$. Now using compositionality of the z-normal forms as in Proposition 4.1 of \cite{HW15} it follows that the formula $\Phi' = (A\cdot z+B)\cdot(C \cdot z+D)^{-1}$ is equivalent to $\Phi$ where $A, B, C, D$ are $A_1\cdot A_2 + B_1 \cdot C_2$, $A_1\cdot B_2 + B_1 \cdot D_2$, $C_1\cdot A_2 + D_1 \cdot C_2$ and $C_1\cdot B_2 + D_1 \cdot D_2$ respectively. Also, clearly the depth of $A, B, C, D$ is at most $c \log_2 \frac{5s}{6} + b+2$. By choosing $c\geq \frac{2}{(\log_2 6 - \log_2 5)}$ it follows that the depths of formulas $A, B, C, D$ are at most $c \log_2 s+ b$. 
To complete the inductive proof we need to prove that $\hat{C}\hat{\pi}+ \hat{D} \neq 0$ for any formula $\pi$ such that $\Phi_{z\leftarrow \pi}$ is correct. Let $\pi$ be such that $\Phi_{z\leftarrow \pi}$ is correct. For simplicity lets denote formulas $\Phi_v$ and $\Phi_{v\leftarrow z'}$ by $\alpha$ and $\beta$ respectively. Since $\Phi_{z\leftarrow \pi}$ is correct, it implies $\alpha_{z\leftarrow \pi}$ is correct as $\alpha$ is a subformula of $\Phi$. So by inductive assumption $\hat{C_2} \hat{\pi} + \hat{D_2} \neq 0$. Since $\Phi_{z\leftarrow \pi}$ is correct, $\beta$ being a subformula of $\Phi$ it also implies $\beta_{z'\leftarrow \gamma}$ is correct where $\gamma= (A_2\cdot \pi + B_2)\cdot (C_2 \cdot \pi +D_2)^{-1}$. 
By inductive assumption we get 
\begin{eqnarray*}
\hat{C_1} \hat{\gamma} + D_1 &\neq& 0~ \text{which implies}\\
\hat{C_1} [(\hat{A_2}\cdot \hat{\pi} + \hat{B_2})\cdot (\hat{C_2} \cdot \hat{\pi} +\hat{D_2})^{-1}]+D_1 &\neq& 0~\text{which implies}\\
\hat{C_1} (\hat{A_2}\cdot \hat{\pi} + \hat{B_2})+ \hat{D_1} (\hat{C_2} \cdot \hat{\pi} +\hat{D_2}) &\neq& 0~\text{as}~\hat{C_2} \hat{\pi} + \hat{D_2} \neq 0\\
\text{So}~(\hat{C_1}\hat{A_2}+ \hat{D_1}\hat{C_2})\hat{\pi} + (\hat{C_1}\hat{B_2}+\hat{D_1}\hat{D_2}) &\neq& 0~\text{which implies}\\
\hat{C}\hat{\pi}+\hat{D} &\neq& 0 
\end{eqnarray*}     
This completes the proof for case 1. 

Next we argue that case 1 and case 2 together cover all the
possibilities. To see this, we will argue that if there does not exist a unique $u_i$ with $wt(u_i)\geq \frac{s}{6}$ then there must exist a balanced gate. There are two possibilities: either for every
gate $u_i$, $wt(u_i)\leq \frac{s}{6}$ or there are two or more gates
$u_i$'s with $wt(u_i)> \frac{s}{6}$ If for every $i\in [\ell]$,
$wt(u_i) \leq \frac{s}{6}$, we find smallest $j$ such that
$\sum_{i=1}^j wt(u_i)>\frac{s}{6}$. Clearly
$\sum_{i=1}^j wt(u_i)\leq \frac{2s}{6}$, which implies
$wt(\Phi_v), wt(\Phi_{v\leftarrow z'}) \leq \frac{5s}{6}$ for
$v= v_{i+1}$. So $v$ is balanced. If there are two or more $u_i$'s
such that $wt(u_i)> \frac{s}{6}$ then $v$ be parent of gate $u_i$ such
that $i$ is a largest index with $wt(u_i)> \frac{s}{6}$. Clearly $v$
is a balanced gate. This proves that case 1, 2 together cover all
possibilities.

Assume that there is a unique $i\in[\ell]$ such that $wt(u_i)> \frac{s}{6}$. We first apply Depth-Reduce on formula $\Phi_{u_i}$ and get a log-depth formula $\Psi_3$ equivalent to $\Phi_{u_i}$, we carry out this depth reduction to efficiently test if $\Phi_{u_i}\equiv 0$?. We will give details on this later when we figure out the parallel time complexity of the algorithm. Now when $\Phi_{u_i}\equiv \Psi_3 \equiv 0$, clearly formula $\Phi\equiv \Phi''$ where $\Phi''$ is a formula obtained from $\Phi$ by replacing sub-formula rooted at $u_i$ by $0$. As $wt(u_i)> \frac{s}{6}$, we have $|\Phi''|\leq \frac{5s}{6}$. So by inductive assumption, the correct sub-formulas $A, B, C, D$  of $\Psi_4$ obtained by recursive call $\text{Normal-Form}(\Phi'',z)$ have depth at most $c\log_2 \frac{5s}{6}+b \leq c\log_2 s+b$ and $\Psi_4 \equiv \Phi''\equiv \Phi$. So it follows, $\Phi_{z\leftarrow \pi} \equiv \Phi''_{z\leftarrow \pi}$. Consequently $\Phi_{z\leftarrow \pi}$ is correct iff $\Phi''_{z\leftarrow \pi}$ is correct. So from inductive hypothesis it follows that $\hat{C} \hat{\pi} + \hat{D} \neq 0$ for any formula $\pi$ such that $\Phi_{z\leftarrow \pi}$ is correct. This proves the correctness of Normal-form procedure when $\Phi_{u_i}\equiv 0$.    

 
Now let $\Phi_{u_i}\not \equiv 0$. Let $v$ be the parent of $u_i$. 
Below we discuss the case when $v$ is $\times$-gate and $u_i$ is a right child of $v$. 

We have $\Psi_2 \equiv \Phi_{v_i}\equiv (A_2\cdot z +B_2)\cdot (C_2\cdot z + D_2)^{-1}$. Let $h_1=A_2\cdot z+B_2$ and $h_2= C_2 \cdot z+ D_2)$. So $\Phi_{v_i}\equiv h_1 \cdot h_2^{-1}$. Now as $v$ is $\times$-gate and $v_i,u_i$ are left and right children of $v$ respectively. So we get $\Phi_v \equiv h_1h_2^{-1}\Phi_{u_i} \equiv h_1h_2^{-1}\Psi_3$. We have $\Psi_1 \equiv \Phi_{v\leftarrow z'}=(A_1z'+B_1)(C_1z'+D_1)^{-1}$. So we get
\begin{eqnarray*}
\Phi &\equiv& (~A_1 \cdot (h_1h_2^{-1}\Psi_3)+B_1)\cdot (~C_1 \cdot(h_1h_2^{-1}\Psi_3)+D_1)^{-1}\\
&\equiv& (A_1 \cdot h_1+B_1\cdot \Psi_3^{-1}h_2)\cdot h_2^{-1}\Psi_3 \cdot [~(C_1\cdot h_1+D_1 \cdot \Psi_3^{-1}h_2)\cdot h_2^{-1}\Psi_3~]^{-1}\\
&\equiv& (A_1 \cdot h_1+B_1 \cdot \Psi_3^{-1}h_2)\cdot h_2^{-1}\Psi_3 \cdot \Psi_3^{-1}h_2~ \cdot (C_1\cdot h_1+D_1\cdot \Psi_3^{-1}h_2)^{-1}\\
&\equiv& (A_1 \cdot h_1+B_1 \cdot \Psi_3^{-1}h_2)\cdot(C_1 \cdot h_1+D_1\cdot \Psi_3^{-1}h_2)^{-1}
\end{eqnarray*}
By substituting values of $h_1, h_2$ and simplifying we get that $\Phi \equiv (A\cdot z+B)\cdot(C \cdot z+D)^{-1}$ where $A, B, C, D$ are $A_1\cdot A_2+B_1 \cdot \Psi_3^{-1}\cdot C_2$, $A_1\cdot B_2+B_1 \cdot \Psi_3^{-1}\cdot D_2$, $C_1\cdot A_2+D_1\cdot \Psi_3^{-1}\cdot C_2$ and $C_1\cdot B_2+D_1 \cdot \Psi_3^{-1} \cdot D_2$ respectively as defined in Step 2(d). 

As $wt(u_i)> \frac{s}{6}$, clearly $|\Phi_{v_i}|, |\Phi_{v\leftarrow z'}| \leq \frac{5s}{6}$. 
Since
\begin{eqnarray*}
\Psi_1&=&(A_1z'+B_1)(C_1z'+D_1)^{-1}=\text{Normal-Form}(\Phi_{v \leftarrow z'},z')\\
\Psi_2&=&(A_2z+B_2)(C_2z+D_2)^{-1}=\text{Normal-Form}(\Phi_{v_i},z)
\end{eqnarray*}

So by inductive assumptions the sub-formulas $A_1, B_1, C_1, D_1$ of $\Psi_1$ and the sub-formulas $A_2, B_2, C_2, D_2$ of $\Psi_2$ are correct and have depths at most $c \log_2 \frac{5s}{6}+b$. As $\Psi_3 = \text{Depth-Reduce}(\Phi_{u_i})$ and $|\Phi_{u_i}|< s$ by inductive assumption we get that the depth of $\Psi_3$ is at most $c\log_2 s$. Clearly $c \log_2 \frac{5s}{6}+b \leq c \log_2 s$ for $c\geq \frac{b}{\log_2 6 - \log_2 5)}$. So from expressions for $A, B, C, D$ it follows that the depth of $A, B, C, D$ is at most depth of $\Psi_3$ plus $4$. Which implies that the depth of $A, B, C, D$ is at most $c \log_2 s + 4 \leq c \log_2 s + b$ if the constant $b\geq 4$. This gives us the desired bound on the depth of $A, B, C, D$.
To summarize if we choose constant $b\geq 4$ and choose constant $c$ such that it satisfies all the lower bounds required in different steps of the above proof, we will get the desired bound on the depth of $A, B, C, D$. We can show that $\hat{C}\hat{\pi}+\hat{D} \neq 0$ for any formula $\pi$ such that $\Phi_{z\leftarrow \pi}$ is correct. The proof is similar to one for Case (1), we additionally need to take into account $\times$-gate at $v$ while composing z-Normal forms $\Psi_1$ and $\Psi_2$. We skip the details. 

When $v$ is a $v$ is a $\times$-gate and $u_i$ is a left child of $v$, the composition of z-normal forms is easy as we do not need an oracle access for $\RIT$ as in the case discussed above when $u_i$ is the right child. In the commutative case we can by default assume that $u_i$ is the left child. Precisely for this reason Brent's construction \cite{Br74} is independent of whether $\Phi_{u_i} \equiv 0$ or not.
 We skip the details of cases when $v$ is a $+$-gate or $v$ is
 $\times$-gate and $u_i$ is the left child which can be handled
 similar to case (1). This proves the correctness of the procedure
 Normal-Form.
 

 Next we analyze the parallel time complexity of both the
 procedures. Let $t_1(s), t_2(s)$ denote the parallel time
 complexities of the procedures Depth-Reduce and Normal-Form
 respectively. The step (1) of the Depth reduce procedure can be
 implemented in $\NC$ so it has parallel time complexity
 $(\log_2 s)^k$ for some absolute constant $k$. As both the formulas
 $\Phi_v$ and $\Phi_{v\leftarrow z}$ have sizes at most $2s/3$, the
 parallel time complexity of step (2) is at most
 $\max(t_1(2s/3), t_2(2s/3))$. So we get the following recurrence for
 $t_1(s)$.
\[ t_1(s) \leq (\log_2 s)^a + \max\left(~t_1 \left( \frac{2s}{3}\right)~,~ t_2\left(\frac{2s}{3}~\right)\right)\]
where $a$ is an absolute constant.
Now we obtain recurrence for $t_2(s)$. In Case (1) of Normal-Form when there exist a balanced gate $v_i$, we can find such a gate in $\NC$. Both the formulas $\Phi_v$ and $\Phi_{v\leftarrow z'}$ have the sizes at most $5s/6$ so step 1(a) takes parallel time $\max(t_1(5s/6), t_2(5s/6))$. In case (2) the formulas $\Phi_{v\leftarrow z'}$, $\Phi_{v_i}$ and $\Phi''$ all have sizes at most $5s/6$ and formula $\Phi_{u_i}$ has size at most $s-1$. So collectively we get the following recurrence for $t_2(s)$
\begin{eqnarray*}
 t_2(s) &\leq& (\log_2 s)^{b} + \max \left( ~t_1 \left( \frac{5s}{6}\right)~,~ t_2\left(\frac{5s}{6}~\right)~,~  t_1(s-1) \right)\\
 &\leq& (\log_2 s)^{b} + \max \left( ~ t_2\left(\frac{5s}{6}~\right)~,~  t_1(s) \right)
\end{eqnarray*}
for an absolute constant $b$. The upper bound $t_1(s), t_2(s) \leq (\log_2 s)^c$ for sufficiently large constant $c$ follows from an easy induction. Hence both the procedures can be implemented in deterministic $\NC$.

\end{proof}

\end{document}